\documentclass[twocolumn,journal]{IEEEtran}
\usepackage[T1]{fontenc}
\usepackage[latin9]{inputenc}
\usepackage{prettyref}
\usepackage{float}
\usepackage{calc}
\usepackage{amsmath}
\usepackage{amssymb}
\usepackage{amsthm}
\usepackage{acronym}
\usepackage{color}

\newtheorem{theorem}{Theorem}
\newtheorem{lemma}{Lemma}
\newtheorem{prop}{Proposition}
\newtheorem{definition}{Definition}

\usepackage{graphicx}
\PassOptionsToPackage{normalem}{ulem}
\usepackage{ulem}
\usepackage[unicode=true,
 bookmarks=true,bookmarksnumbered=true,bookmarksopen=true,bookmarksopenlevel=1,
 breaklinks=false,pdfborder={0 0 0},pdfborderstyle={},backref=false,colorlinks=false]
 {hyperref}
\hypersetup{pdftitle={Your Title},
 pdfauthor={Your Name},
 pdfpagelayout=OneColumn, pdfnewwindow=true, pdfstartview=XYZ, plainpages=false}

\makeatletter

\floatstyle{ruled}
\newfloat{algorithm}{tbp}{loa}
\providecommand{\algorithmname}{Algorithm}
\floatname{algorithm}{\protect\algorithmname}

 \let\oldforeign@language\foreign@language
 \DeclareRobustCommand{\foreign@language}[1]{%
   \lowercase{\oldforeign@language{#1}}}

\usepackage[caption=false,font=footnotesize]{subfig}

\@ifundefined{showcaptionsetup}{}{%
 \PassOptionsToPackage{caption=false}{subfig}}
\usepackage{subfig}

\usepackage{cite}

\makeatother

\acrodef{SPD}{Symmetric and Positive Definite}
\acrodef{PT}{Parallel Transport}

\begin{document}

\title{Parallel Transport on the Cone Manifold of SPD Matrices for Domain Adaptation}

\author{Or~Yair,~\IEEEmembership{Student Member,~IEEE,} Mirela Ben-Chen,
and~Ronen~Talmon,~\IEEEmembership{Member,~IEEE}\thanks{The authors are with the Viterbi Faculty of Electrical Engineering, Technion-Israel Institute of Technology, Haifa 32000, Israel (e-mail: \protect\href{http://oryair@campus.technion.ac.il}oryair@campus.technion.ac.il; \protect\href{http://ronen@ee.technion.ac.il}{ronen@ee.technion.ac.il}).}
}



\maketitle
\begin{abstract}
In this paper, we consider the problem of domain adaptation.
We propose to view the data through the lens of covariance matrices and present a method for domain adaptation using parallel transport on the cone manifold of symmetric positive-definite matrices. We provide rigorous analysis using Riemanninan geometry, illuminating the theoretical guarantees and benefits of the presented method. In addition, we demonstrate these benefits using experimental results on simulations and real-measured data.

\end{abstract}


\IEEEpeerreviewmaketitle{}

\section{Introduction}

\IEEEPARstart{T}{he} increasing technological sophistication of current data acquisition systems gives rise to complex, multimodal datasets in high-dimension.
As a result, the acquired data do not live in a Euclidean space, and applying analysis and learning algorithms directly to the data often leads to subpar performance.


To facilitate the analysis and processing of such data, one approach is to observe complex high-dimensional data through the lens of objects with a \emph{known} non-Euclidean geometry. Notable examples of such objects are \ac{SPD} matrices, which live on a cone manifold with a Riemannian metric. One of the most common forms of \ac{SPD} matrices is a covariance matrix, which captures the linear relations between the different data coordinates. These relations are typically simple to compute, and therefore, recently, have become popular features in many applications in computer vision, medical imaging, and machine learning \cite{sra2015conic,tuzel2008pedestrian,freifeld2014model,bergmann2017priors}. In particular, in \cite{pennec2006riemannian,barachant2013classification}, the Riemannian geometry of covariance matrices was studied and exploited for medical imaging and physiological signal analysis.

Typically, Riemannian geometry is used to map objects from the non-Euclidean manifold to a linear Euclidean space by projection onto a tangent plane of the manifold. In existing work, the use of Riemannian geometry is usually limited to a single tangent plane. This indicates a hidden assumption that the \ac{SPD} matrices corresponding to the data are confined to a local region of the manifold. However, the \ac{SPD} matrices of the data often do not live in a small neighborhood on the manifold, and thus the resulting calculations may be inaccurate. 

One such particular scenario, in which \ac{SPD} matrices span a large portion of the cone manifold, occurs when the data comprise multiple domains corresponding to multiple sessions, subjects, batches, etc.
For example, we will show that in a Brain-Computer-Interface (BCI) experiment, the covariance matrices of data acquired from a single subject in a specific session capture well the overall geometric structure of the data. Conversely, when the data consist of measurements from several subjects or several sessions, then the covariance matrices do not live in the same region of the manifold.

Often, multi-domain data pose significant challenges to learning approaches. For example, in the BCI experiment, it is challenging to train a classifier based on data from one subject (session) and apply it to data from another subject (session). This problem is largely referred to as \emph{domain adaptation} or \emph{transfer learning}, and it has attracted a significant research effort in recent years \cite{ben2007analysis,pan2011domain}. 

Broadly, in domain adaptation, the main idea is to adapt a given model that is well performing on a particular domain, to a different yet related domain \cite{ben2007analysis,pan2011domain}. Specifically, in the context of the cone manifold of \ac{SPD} matrices, previous work proposed (geometric) mean subtraction as a simple method for domain adaption of BCI data \cite{barachant2013classification}. Although this approach provided reasonable results for overcoming the differences between multiple sessions of a single subject, we show here that it fails to overcome the differences between multiple subjects. In
\cite{freifeld2014model}, a \ac{PT} approach was proposed, which can be applied either directly to the data, or to a generative model of the data to reduce the computational load for large datasets. However, their approach considers a general Riemannian manifold. Since there is no closed-form expression of \ac{PT} on Riemannian manifolds, besides the sphere manifold and the manifold of all \ac{SPD} matrices \cite{sra2015conic,bergmann2017priors}, no specific scheme or algorithm was provided. We note that \ac{PT} can be approximated using Schild's Ladder \cite{lorenzi2011schild}, an approach that has been used extensively on the manifold of imaging data \cite{pennec2006riemannian,kim2014canonical,bergmann2017priors,lorenzi2011schild}.  

In this paper, we propose a domain adaptation method using the analytic expression of \ac{PT} on the cone manifold of \ac{SPD} matrices. 
We claim that this is a natural and efficient solution for domain adaptation, which enjoys several important benefits. First, the solution is especially designed for \ac{SPD} matrices, which have proven to be good features of data in a gamut of previous work \cite{pennec2006riemannian,barachant2013classification,kim2014canonical}. Second, the analytic form of \ac{PT} on the cone manifold circumvents approximations. Third, \ac{PT} can be efficiently implemented, in contrast to the computationally demanding Schild's Ladder approximation. 
We establish the mathematical foundation of the proposed domain adaptation method. To this end, we provide new results in the geometry of \ac{SPD} matrices. In addition, we show applications to both simulation and real recorded data, obtaining improved performance compared to the competing methods.

In parallel to our study, recent work \cite{zanini2017transfer} has proposed a scheme for transfer learning using the Riemannian geometry of \ac{SPD} matrices, with a tight connection to the present work. We will show that the affine transformation proposed in \cite{zanini2017transfer} can be recast as \ac{PT}. In this paper we provide the mathematical foundation to analyze this transport, we discuss the advantage of our solution compared to  \cite{zanini2017transfer}, and we point out the special case in which the two methods coincide.

This paper is organized as follows. In Section \ref{sec:prelim}, we present preliminaries on the Riemannian geometry of \ac{SPD} matrices. In Section \ref{sec:DomainAdaptaion}, we formulate the problem, present the proposed domain adaptation method, and provide mathematical analysis and justification. Section \ref{sec:Results}
shows experimental results on both simulation and real data. Finally, we conclude the paper in Section \ref{sec:Conclusions}.

\section{Preliminaries on Riemannian Geometry of SPD Matrices}
\label{sec:prelim}

In this section we provide the preliminaries regarding \ac{SPD} matrices, and we refer the reader to the book \cite{bhatia2009positive} for a detailed exposition of this topic. We note that in this paper we focus on covariance matrices, however the statements also hold for general \ac{SPD} matrices. By definition, an \ac{SPD}  matrix $\boldsymbol{P}\in \mathbb{R}^{n\times n}$ has only strictly positive eigenvalues. An alternative definition is that for any vector $\boldsymbol{v} \neq  \boldsymbol{0}$ the quadratic form is strictly positive, i.e., $\boldsymbol{v}^T\boldsymbol{P}\boldsymbol{v}>0$.  

\subsection{Metric and distance}
The definition of an \ac{SPD} matrix entails that the collection of all \ac{SPD} matrices constitutes a convex half-cone in the vector space of real $n\!\times\!n$ symmetric matrices.
This cone forms a differentiable Riemannian manifold $\mathcal{M}$ equipped with the following inner product 
\begin{equation}\label{eq:inner_product}
\big\langle \boldsymbol{S}_{1},\boldsymbol{S}_{2}\big\rangle _{\mathcal{T}_{\boldsymbol{P}}\mathcal{M}}=\big\langle \boldsymbol{P}^{-\frac{1}{2}}\boldsymbol{S}_{1}\boldsymbol{P}^{-\frac{1}{2}},\boldsymbol{P}^{-\frac{1}{2}}\boldsymbol{S}_{2}\boldsymbol{P}^{-\frac{1}{2}}\big\rangle, 
\end{equation}
where  $\mathcal{T}_{\boldsymbol{P}}\mathcal{M}$ is the tangent space at the point $\boldsymbol{P}\!\in\!\mathcal{M}$, $\boldsymbol{S}_{1},\boldsymbol{S}_{2}\in\mathcal{T}_{\boldsymbol{P}}\mathcal{M}$, and $\langle \cdot,\cdot\rangle $ is the standard Euclidean inner product operation. The symmetric matrices $\boldsymbol{S}\in\mathcal{T}_{\boldsymbol{P}}\mathcal{M}$ in the tangent plane live in a linear space, and therefore, we can view them as vectors (with a proper representation). Throughout this paper, we interchangeably use the terms vectors and symmetric matrices when referring to $\boldsymbol{S}\in\mathcal{T}_{\boldsymbol{P}}\mathcal{M}$. 

This Riemannian manifold is a Hadamard manifold, namely, it is simply connected and it is a complete Riemannian manifold of non-positive sectional curvature. Manifolds with non-positive curvature have a unique geodesic curve between any two points, a property that will later be exploited.
Specifically, the unique geodesic curve between any two \ac{SPD} matrices $\boldsymbol{P}_{1},\boldsymbol{P}_{2}\in\mathcal{M}$  is given by \cite[Thm 6.1.6]{bhatia2009positive}
\begin{align}
\label{eq:Geodesic}
\varphi(t) & = \boldsymbol{P}_{1}^{\frac{1}{2}}\big(\boldsymbol{P}_{1}^{-\frac{1}{2}}\boldsymbol{P}_{2}\boldsymbol{P}_{1}^{-\frac{1}{2}}\big)^{t}\boldsymbol{P}_{1}^{\frac{1}{2}},\qquad0\leq t\leq1.
\end{align}
 The arc-length of the geodesic curve defines the following Riemannian distance on the manifold \cite{bhatia2009positive}:
\begin{align*}
\label{eq:RiemannianMetric}
d_{R}^{2}\big(\boldsymbol{P}_{1},\boldsymbol{P}_{2}\big) & = \big\Vert \log\big(\boldsymbol{P}_{2}^{-\frac{1}{2}}\boldsymbol{P}_{1}\boldsymbol{P}_{2}^{-\frac{1}{2}}\big)\big\Vert_{F}^{2} \\
& = \textstyle \sum_{i=1}^n\log^{2}\Big(\lambda_{i}\big(\boldsymbol{P}_2^{-\frac{1}{2}} \boldsymbol{P}_{1}\boldsymbol{P}_{2}^{-\frac{1}{2}}\big)\Big)\\
& = \textstyle \sum_{i=1}^n\log^{2}\left(\lambda_{i}\big(\boldsymbol{P}_{1}\boldsymbol{P}_{2}^{-1}\big)\right),
\end{align*}
where $\boldsymbol{P}_{1},\boldsymbol{P}_{2}\in\mathcal{M}$, $\Vert \cdot\Vert _{F}$ is the Frobenius norm, $\log(\boldsymbol{P})$ is the matrix logarithm, and $\lambda_{i}(\boldsymbol{P})$ is the $i$-th eigenvalue of $\boldsymbol{P}$. We additionally denote by $\tfrac{d}{dt}\varphi(t)=\varphi'(t)\!\in\!\mathcal{T}_{\varphi(t)}\mathcal{M}$ the velocity vector of the geodesic at $t\!\in\![0,1]$.
Figure \ref{fig:2X2_Cone_Example} presents an illustration of the geodesic curve and the Riemannian distance. The cone manifold of $2\times 2$ \ac{SPD} matrices can be displayed in $\mathbb{R}^3$, since any symmetric matrix $\boldsymbol{P}=\bigl(\begin{smallmatrix}x & y\\
y & z \end{smallmatrix}\bigr)$ is positive if and only if $x>0, z>0$ and $y^2<xz$.

\begin{figure}[t]
\centering \includegraphics[width=1\columnwidth]{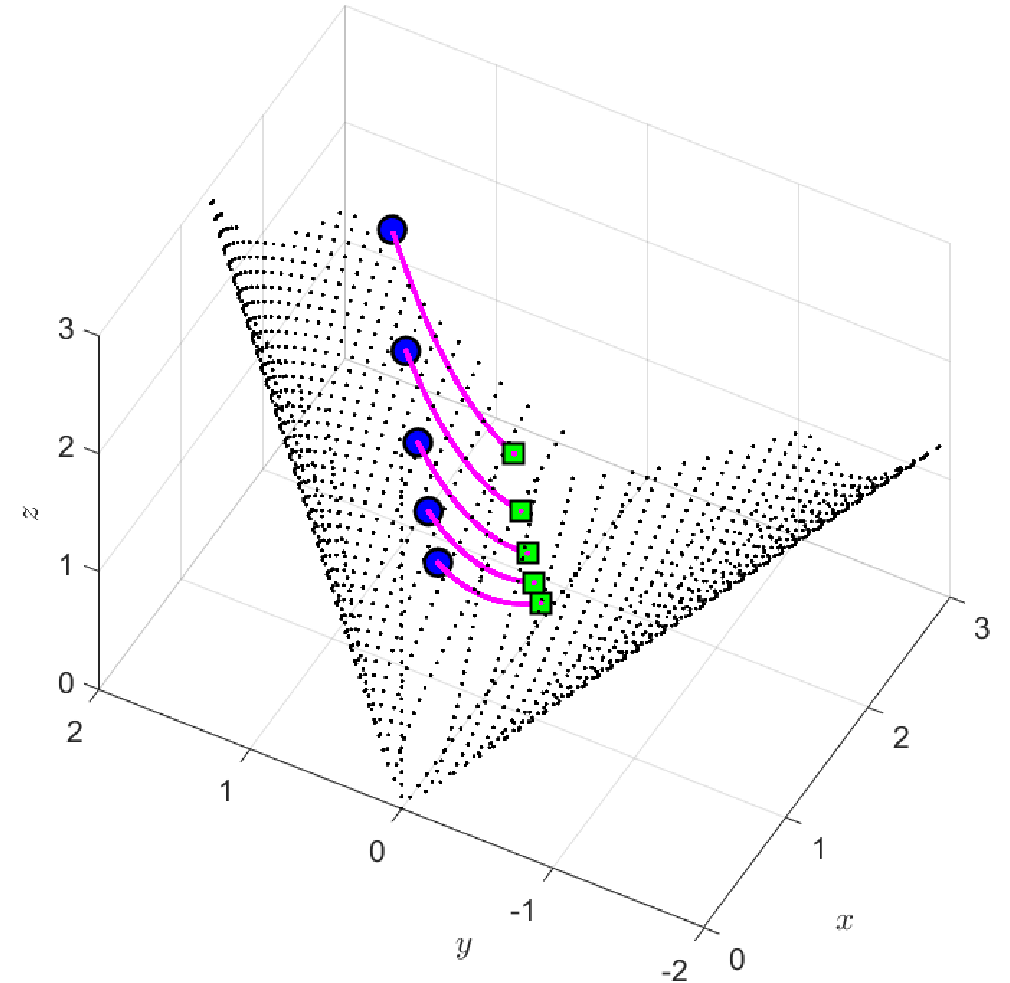}
\caption{The cone manifold of $2\times2$ SPD matrices. The black dots mark the boundary of the cone (i.e., matrices with eigenvalue zero). Each magenta curve is the geodesic between pairs of matrices (blue circles and green squares). All the geodesic curves are of the same length (i.e., the Riemannian distance between all the pairs is equal).}
\label{fig:2X2_Cone_Example}
\end{figure}

\subsection{Exponential and Logarithm maps}
The Logarithm map, which projects an \ac{SPD} matrix $\boldsymbol{P}_i\!\in\!\mathcal{M}$ to the tangent plane $\mathcal{T}_{\boldsymbol{P}}\mathcal{M}$ at $\boldsymbol{P}\!\in\!\mathcal{M}$, is given by 

\begin{align*}
\boldsymbol{S}_{i}  =\text{Log}_{\boldsymbol{P}}(\boldsymbol{P}_{i})=\boldsymbol{P}^{\frac{1}{2}}\log\big(\boldsymbol{P}^{-\frac{1}{2}}\boldsymbol{P}_{i}\boldsymbol{P}^{-\frac{1}{2}}\big)\boldsymbol{P}^{\frac{1}{2}} \in \mathcal{T}_{\boldsymbol{P}}\mathcal{M}.
\end{align*}
The Exponential map, which projects a vector $\boldsymbol{S}_i\!\in\!\mathcal{T}_{\boldsymbol{P}}\mathcal{M}$ back to the manifold $\mathcal{M}$ is given by
\begin{align}
\boldsymbol{P}_{i} =\text{Exp}_{\boldsymbol{P}}(\boldsymbol{S}_{i})=\boldsymbol{P}^{\frac{1}{2}}\exp\big(\boldsymbol{P}^{-\frac{1}{2}}\boldsymbol{S}_{i}\boldsymbol{P}^{-\frac{1}{2}}\big)\boldsymbol{P}^{\frac{1}{2}} \in \mathcal{M}.
\end{align}
An important property relates the Logarithm and Exponential maps to the geodesic curve.
Formally, let $\boldsymbol{P}_{1},\boldsymbol{P}_{2}\!\in\!\mathcal{M}$, and consider the (unique) geodesic $\varphi(t)$ from $\boldsymbol{P}_1$ to $\boldsymbol{P}_2$. The initial velocity $\varphi'(0)\! \in \! \mathcal{T}_{\boldsymbol{P}_1}\mathcal{M}$ is given by the Logarithm map $\varphi'\left(0\right)\!=\!\text{Log}_{\boldsymbol{P}_{1}}\left(\boldsymbol{P}_{2}\right)$.
Similarly, the Exponential map projects the initial velocity vector  $\varphi'\left(0\right)$ back to $\boldsymbol{P}_2$, namely, $\boldsymbol{P}_{2}=\text{Exp}_{\boldsymbol{P}_{1}}\left(\varphi'\left(0\right)\right)$.

\subsection{Riemannian mean}

The Riemannian mean $\overline{\boldsymbol{P}}$ of a set $\left\{ \boldsymbol{P}_{i} | \boldsymbol{P}_{i}\in\mathcal{M}\right\} $ is defined using the Fr\'echet mean:
\begin{equation}\label{eq:riemannian_mean}
\overline{\boldsymbol{P}}\triangleq\arg\min_{\boldsymbol{P}\in\mathcal{M}}\sum_{i}d_{R}^{2}\big(\boldsymbol{P},\boldsymbol{P}_{i}\big).
\end{equation}
A special case is the Riemannian mean $\overline{\boldsymbol{P}}$ of two \ac{SPD} matrices $\boldsymbol{P}_{1},\boldsymbol{P}_{2}\!\in\!\mathcal{M}$, which has a closed-form expression, and is located at the midpoint of the geodesic curve:
\[
\overline{\boldsymbol{P}}=\varphi(\tfrac{1}{2})=\boldsymbol{P}_{1}^{\frac{1}{2}}\big(\boldsymbol{P}_{1}^{-\frac{1}{2}}\boldsymbol{P}_{2}\boldsymbol{P}_{1}^{-\frac{1}{2}}\big)^{\frac{1}{2}}\boldsymbol{P}_{1}^{\frac{1}{2}}.
\]
Generally, for more than two matrices, the solution of the optimization problem \eqref{eq:riemannian_mean} can be obtained by an iterative procedure. Barachant \emph{et al} \cite{barachant2013classification} presented an algorithm based on \cite{moakher2005differential} for estimating the Riemannian mean. For completeness, we include their algorithm in Appendix \ref{alg:RiemannianMean}.

Given a set $\left\{ \boldsymbol{P}_{i}|\boldsymbol{P}_{i}\in\mathcal{M}\right\} $ and its Riemannian mean $\overline{\boldsymbol{P}}$, there is a commonly used approximation of the Riemannian distances on $\mathcal{M}$ in the neighborhood of $\overline{\boldsymbol{P}}$. 
Specifically, the approximation of the Riemannian distance $d_{R}^{2}$ is given by:
\begin{align}
\label{eq:TangentApproxiamtion}
d_{R}^{2}\big(\boldsymbol{P}_{i},\boldsymbol{P}_{j}\big)\approx\big\Vert \hat{\boldsymbol{S}}_{i}-\hat{\boldsymbol{S}}_{j}\big\Vert _{F}^{2}\,\,,
\end{align}
where $\hat{\boldsymbol{S}}_{i}=\overline{\boldsymbol{P}}^{-\tfrac{1}{2}}\text{Log}_{\overline{\boldsymbol{P}}}(\boldsymbol{P}_{i})\overline{\boldsymbol{P}}^{-\tfrac{1}{2}}$.
For more details on the accuracy of this approximation, see \cite{tuzel2008pedestrian}.

\section{Domain Adaptation with Parallel Transport}
\label{sec:DomainAdaptaion}
\subsection{Overview}
Let $\mathcal{X}^{(1)}\!=\!\big\{ \boldsymbol{x}_{i}^{(1)}(t)\big\} _{i=1}^{N_{1}}$ and $\mathcal{X}^{(2)}=\big\{ \boldsymbol{x}_{i}^{(2)}(t)\big\} _{i=1}^{N_{2}}$ be two subsets of $N_1$ and $N_2$ high-dimensional time series, respectively, where $\boldsymbol{x}_i^{(k)}(t)\in \mathbb{R}^D$. Suppose each subset lives in a particular domain, which could be related to the acquisition modality, session, deployment, and set of environmental conditions. In our notation, the superscript $k$ denotes the index of the subset, the subscript $i$ denotes the index of the time-series within each subset, and $t$ represents the time axis of each time-series. 

Our exposition focuses only on two subsets, and the generalization for any number of subsets is discussed at the end of this section. In addition, we consider here time-series, but our derivation does not take the temporal order into account, and therefore, the extension to other types of data, where $t$ is merely a sample index, e.g., images, is straight-forward.

Analyzing such data typically raises many challenges. For example, a long-standing problem is how to efficiently compare between high-dimensional point clouds, and particularly, time-series. When the data are measured signals, sample comparisons become even more challenging, since such high-dimensional measured data usually contain high levels of noise.

In particular, in our setting, we face an additional challenge, since the data is given in different domains; comparing time-series from the same subset is a difficult task by itself, even more so is comparing time-series from two subsets from different domains.

Our goal is to find a new \emph{joint} representation of the two subsets in an unsupervised manner. Broadly, we aim to devise a low-dimensional representation in a Euclidean space that facilitates efficient and meaningful comparisons. 
For the purpose of evaluation, we associate the time-series $\boldsymbol{x}_i^{(k)}(t)$ with labels $y_i^{(k)}$ and define ``meaningful'' comparisons with respect to these labels. 
More concretely, we evaluate the joint representation by the Euclidean distance between the new representation of any two time-series with similar corresponding labels, independently of the time-series respective domain.
We emphasize that the proposed approach is unsupervised and it does not depend on the labels, which are only considered for the purpose of evaluation.

Devising such a new representation will facilitate efficient and accurate domain adaptation schemes. 
Specifically, given a subset $\big\{ \boldsymbol{x}_{i}^{\left(1\right)}\left(t\right)\big\} _{i=1}^{N_{1}}$  with corresponding labels $\big\{ y_{i}^{\left(1\right)}\big\} _{i=1}^{N_{1}}$, we could train a classifier based on the new derived representation of the subset. Then, when another unlabeled subset $\big\{ \boldsymbol{x}_{i}^{\left(2\right)}\left(t\right)\big\} _{i=1}^{N_{2}}$ becomes available, we could apply the trained classifier to the derived (joint) representation of the latter subset.

\subsection{Illustrative example}

To put the problem setting and our proposed solution in context, throughout the paper, we will follow an illustrative example, taken from the brain computer interface (BCI) competition IV (dataset IIa) \cite{schlogl200719}.
Consider data from a BCI experiment of motor imagery comprising of recordings from $D=22$ Electroencephalography (EEG) electrodes. The dataset contains several subjects, where each subject was asked repeatedly to perform one out of four motor imagery tasks: movement of the right hand, the left hand, both feet, and the tongue.

Let $\mathcal{X}^{(1)}\!=\!\big\{ \boldsymbol{x}_{i}^{\left(1\right)}\left(t\right)\big\} _{i=1}^{N_{1}}$ be a subset of recordings acquired from a single subject, indexed $(1)$, where the time-series $\boldsymbol{x}_i^{(1)}(t)$ consists of the signals, recorded simultaneously from the $D$ EEG channels during the $i$-th repetition/trial. Each time series $\boldsymbol{x}_i^{(1)}(t)$ is attached with a label $y_i^{(1)}$, denoting the imagery task performed at the the $i$-th trial.
Common practice is to train a classifier based on $\mathcal{X}^{(1)}$, so that the imagery task could be identified from new EEG recordings. This capability could then be the basis for devising brain computer interfaces, for example, to control prosthetics. 

Suppose a new subset $\mathcal{X}^{(2)}\!=\!\big\{ \boldsymbol{x}_{i}^{\left(2\right)}\left(t\right)\big\} _{i=1}^{N_{2}}$ of recordings acquired from another subject, indexed $(2)$, becomes available. Applying the classifier, trained based on data from subject $(1)$, to the new subset of recordings from subject $(2)$ yields poor results, as we will demonstrate in Section \ref{sub:BCI}. Indeed, most methods addressing this particular challenge, as well as related problems, exclusively analyze data from each individual subject separately. By constructing a joint representation for both $\mathcal{X}^{(1)}$ and $\mathcal{X}^{(2)}$, which is oblivious to the specific subject, we develop a classifier that is trained on data from one subject and applied to data from another subject without any calibration, i.e., without any labeled data from the new (test) subject.

\subsection{Covariance matrices as data features}

As described before, we suggest looking at the data through the lens of covariance matrices. We denote the covariance matrices by: 
\begin{equation*}
\boldsymbol{P}_{i}^{(k)}=\mathbb{E}\Big[\big(\boldsymbol{x}_{i}^{(k)}(t)-\boldsymbol{\mu}_{i}^{(k)}\big)\big(\boldsymbol{x}_{i}^{(k)}(t)-\boldsymbol{\mu}_{i}^{(k)}\big)^{T}\Big]\,,
\end{equation*}
where $\boldsymbol{\mu}_i^{(k)} = \mathbb{E}\big[\boldsymbol{x}_{i}^{(k)}(t)\big]$. Typically, since the statistics of the data is unknown, we use estimates of the covariance, such as the sample covariance. We note that our approach is applicable to any kind of input data given as \ac{SPD} matrices. For example, in machine learning, common practice is to use kernels which represent an inner product between features after some non-linear transformation \cite{hofmann2008kernel}. 


By using covariance matrices as data features we enjoy a few key benefits. First, since covariance matrices are computed from data by averaging over time, they tend to be robust to noise. Second, covariance matrices can be seen as a low dimensional representation. Third, they have useful geometric properties and a well-developed Riemannian framework, as described in Section \ref{sec:prelim}. Particularly, they have a Riemannian metric \eqref{eq:RiemannianMetric}, facilitating appropriate data samples comparisons, which is a basic ingredient of many analysis and learning techniques.
In this work, we build on and extend the latter.

Recently, the usefulness of covariance matrices has been demonstrated in the context of the BCI problem \cite{barachant2013classification}.  
There, Barachant et al. considered data from a single subject and proposed to project the covariance matrices $\left\{ \boldsymbol{P}_{i}\right\} $ of the recordings from each trial (after some whitening) into the tangent plane of the Riemannian mean $\overline{\boldsymbol{P}}$, namely compute $\boldsymbol{S}_{i}\!=\!\text{Log}_{\overline{\boldsymbol{P}}}\left(\boldsymbol{P}_{i}\right)$. Then, a classifier was trained on the set $\left\{ \boldsymbol{S}_{i}\right\} $. Using this approach, state of the art results for motor imagery task classification were obtained.
However, when considering several subsets from multiple domains, such as different sessions or subjects, as reported in \cite{barachant2013classification}, the covariance matrices convey a domain-specific content, which in turn poses limitations on task classification. 
For multiple sessions on different days, Barachant et al. proposed to subtract the Riemannian mean from each subset, namely, to project each subset $\mathcal{P}^{(k)}=\big\{ \boldsymbol{P}_{i}^{\left(k\right)}\big\} $ to the tangent space at its own mean. Indeed, when the train set and the test set were obtained on different days, this mean normalization improved the task classification rate. However, in the case of multiple subjects, this approach is inadequate.  As mentioned before, given recordings from one subject as a train set and recordings from another subject as a test set, the classification of the different mental tasks based on covariance matrices fails completely. 

This illuminates the primary challenge addressed in this work -- how to build a representation so that any two covariance matrices associated with the same mental task, but from possibly different sessions or subjects, will be given a similar representation.
Importantly, since the task labels are unknown, this objective cannot be directly imposed. In the sequel, we exploit the Riemannian geometry of covariance matrices, and devise such a representation in an unsupervised manner by preserving local geometric structures.

\subsection{Formulation}

Consider two subsets $\mathcal{P}^{(1)}$ and $\mathcal{P}^{(2)}$ from two different domains consisting of $N_1$ and $N_2$ covariance matrices, respectively. Let  $\overline{\boldsymbol{P}}^{(1)}$ and $\overline{\boldsymbol{P}}^{(2)}$ be their respective Riemannian means. 
Let $\varphi(t)$, given explicitly in \eqref{eq:Geodesic}, denote the unique geodesic from  $\overline{\boldsymbol{P}}^{(2)}$ to $\overline{\boldsymbol{P}}^{(1)}$ such that $\varphi(0)\!=\!\overline{\boldsymbol{P}}^{(2)}$ and $\varphi\left(1\right)\!=\!\overline{\boldsymbol{P}}^{\left(1\right)}$.
Finally, let $\boldsymbol{S}_{i}^{\left(k\right)}$ be the symmetric matrix (or equivalently, the vector) in the tangent space $\mathcal{T}_{\overline{\boldsymbol{P}}^{\left(k\right)}}\mathcal{M}$, obtained by projecting $\boldsymbol{P}_{i}^{\left(k\right)}$ to $\mathcal{T}_{\overline{\boldsymbol{P}}^{\left(k\right)}}\mathcal{M}$:
\[
\boldsymbol{S}_{i}^{\left(k\right)}=\text{Log}_{\overline{\boldsymbol{P}}^{\left(k\right)}}\big(\boldsymbol{P}_{i}^{\left(k\right)}\big)\,,
\]
for $k\in \{1,2\}$ and $i\in \{1,2,\dots,N_k\}$ . 

Our goal now is to derive a new representation $\Gamma\big(\boldsymbol{S}_{i}^{(2)}\big)$ of $\boldsymbol{S}_{i}^{(2)}$ given by the map $\Gamma\!:\!\mathcal{T}_{\overline{\boldsymbol{P}}^{(2)}}\mathcal{M}\to\mathcal{T}_{\overline{\boldsymbol{P}}^{(1)}}\mathcal{M}$, such that $\big\{ \boldsymbol{S}_{i}^{(1)} \big\}$ and $\big\{ \Gamma\big(\boldsymbol{S}_{i}^{(2)}\big) \big\}$ live in the same space. 
This allows us to relate samples from the two subsets, and compute quantities such as $\big\langle \boldsymbol{S}_{i}^{(1)},\Gamma\big(\boldsymbol{S}_{j}^{(2)}\big)\big\rangle _{\overline{\boldsymbol{P}}^{(1)}} $. 
In addition, we require that the new representation will fulfill the following properties:

\begin{enumerate}
\item 
Zero mean:
\[
\tfrac{1}{N_{2}}\textstyle\sum_{i=1}^{N_{2}}\Gamma\big(\boldsymbol{S}_{i}^{(2)}\big)=\tfrac{1}{N_{1}}\textstyle\sum_{i=1}^{N_{1}}\boldsymbol{S}_{i}^{(1)}=0
\]

\item 
Inner product preservation:
\[
\big\langle \Gamma\big(\boldsymbol{S}_{i}^{(2)}\big),\Gamma\big(\boldsymbol{S}_{j}^{(2)}\big)\big\rangle _{\overline{\boldsymbol{P}}^{(1)}}=\big\langle \boldsymbol{S}_{i}^{(2)},\boldsymbol{S}_{j}^{(2)}\big\rangle _{\overline{\boldsymbol{P}}^{(2)}}
\]
for all $i,j \in \{1,\dots,N_2\}$.

\item 
Geodesic velocity preservation:
\begin{equation}\label{eq:3prop}
\Gamma\left({\varphi}'\left(0\right)\right)={\varphi}'\left(1\right)
\end{equation}
\end{enumerate}

Properties (1) and (2) imply that the new representation $\Gamma$ preserves inter-sample relations, defined by the inner product. Note that a map $\Gamma$ satisfying properties (1) and (2) is not unique;
for any $\Gamma$ admitting to properties (1) and (2), the composition $R\circ \Gamma$, where $R$ is an arbitrary rotation within the subspace $\mathcal{T}_{\overline{\boldsymbol{P}}^{(1)}}\mathcal{M}$, satisfies properties (1) and (2) as well.
To resolve this arbitrary degree of freedom, we use the geodesic between two points on the \ac{SPD} manifold, which is unique \cite{bhatia2009positive}.
Concretely, in property (3), the two intrinsic symmetric matrices (vectors) ${\varphi}'(0)\!\in\!\mathcal{T}_{\overline{\boldsymbol{P}}^{(2)}}\mathcal{M}$ and ${\varphi}'(1)\!\in\!\mathcal{T}_{\overline{\boldsymbol{P}}^{(1)}}\mathcal{M}$, induced by the velocity of the unique geodesic at the source and destination, are used to fix a rotation and to align the subset $\big\{ \Gamma\big(\boldsymbol{S}_{i}^{(2)}\big)\big\} $ with the subset $\big\{ \boldsymbol{S}_{i}^{(1)}\big\}$.  

We remark that the above properties imply that the subset $\big\{ \Gamma\big(\boldsymbol{S}_{i}^{(2)}\big)\big\}$ is embedded in the $\langle \cdot,\cdot\rangle _{\overline{\boldsymbol{P}}^{(1)}}$ inner product space. In the sequel, we will describe how to circumvent the dependence of the inner product space on $\overline{\boldsymbol{P}}^{(1)}$ and make the new representation truly Euclidean by pre-whitening the data. Additionally, note that the mean subtraction presented in \cite{barachant2013classification} admits only properties (1)-(2). 


\subsection{Domain adaptation}
First, we explicitly provide the expression for parallel transport on the SPD cone manifold, and then we use it to define the map $\Gamma$.
\begin{lemma}[Parallel Transport]
\label{def:PT}
Let $\boldsymbol{A},\boldsymbol{B}\!\in\!\mathcal{M}$. The \ac{PT} from $\boldsymbol{B}$ to $\boldsymbol{A}$ of any $\boldsymbol{S} \!\in \!\mathcal{T}_{\boldsymbol{B}}\mathcal{M}$ is given by:
\begin{equation}\label{eq:PT}
\Gamma_{\boldsymbol{B}\to \boldsymbol{A}}\left({\boldsymbol{S}}\right)\triangleq{\boldsymbol{E}}{\boldsymbol{S}}{\boldsymbol{E}}^{T},
\end{equation}
where ${\boldsymbol{E}}=\big({\boldsymbol{A}}{\boldsymbol{B}}^{-1}\big)^{\tfrac{1}{2}}$.
\end{lemma}
This lemma was presented in \cite[Eq.~3.4]{sra2015conic}. The proof of the lemma is given in Appendix \ref{sec:proofLemma_ESE} and it is based on \cite{ferreira2006newton}.
An illustration of the \ac{PT} on the \ac{SPD} manifold is presented in Figure \ref{fig:Cone2X2_PT}. 
Note that the inner products between the three vectors in the figure are preserved under the parallel transport $\Gamma_{\boldsymbol{B}\to \boldsymbol{A}}$ and the appearance could be misleading since the space is not Euclidean.
\begin{figure}[t]
\centering 
\includegraphics[width=1\columnwidth]{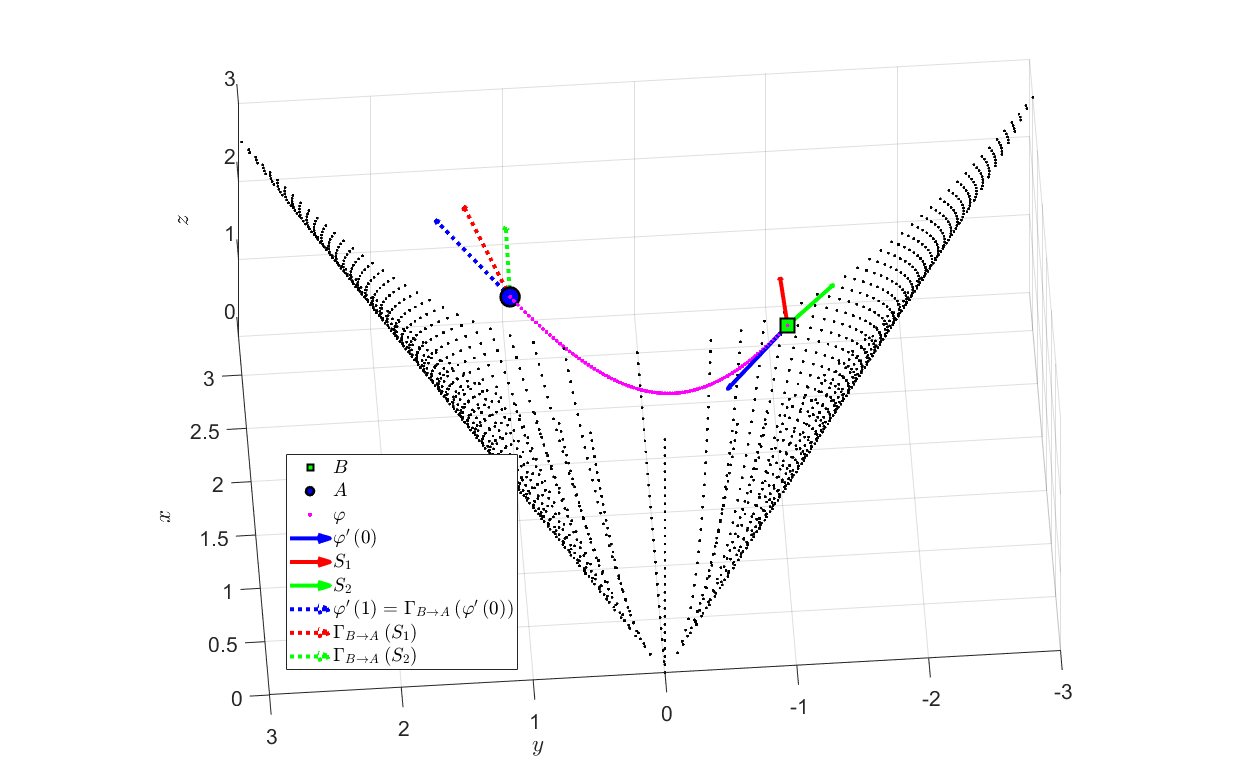}
\caption{Illustration of the \ac{PT} on the \ac{SPD} manifold. $\boldsymbol{A}$ and $\boldsymbol{B}$ are two \ac{SPD} matrices, and $\varphi$ is the unique geodesic between them. We plot three vectors in $\mathcal{T}_{\boldsymbol{B}}\mathcal{M}$: $\varphi'(0)$, $S_1$ and $S_2$ along with their corresponding parallel transported vectors to $\mathcal{T}_{\boldsymbol{A}}\mathcal{M}$ using $\Gamma_{\boldsymbol{B}\to \boldsymbol{A}}$.}
\label{fig:Cone2X2_PT}
\end{figure}

\begin{theorem}
\label{thm:Gamma}
The representation $\Gamma _{\overline{\boldsymbol{P}}^{(2)}\to\overline{\boldsymbol{P}}^{(1)}}\big(\boldsymbol{S}_{i}^{(2)}\big)$ given by the unique \ac{PT} of $\boldsymbol{S}_{i}^{(2)}$ from $\overline{\boldsymbol{P}}^{(2)}$ to $\overline{\boldsymbol{P}}^{(1)}$ is well defined and satisfies properties $(1)-(3)$.
\end{theorem}
The proof is given in Appendix \ref{sec:proofThm_Gamma}.

Theorem \ref{thm:Gamma} sets the stage for domain adaptation. We propose a map $\Psi:\mathcal{M}\to\mathcal{M}$ that adapts the domain of the subset of \ac{SPD} matrices $\mathcal{P}^{(2)}$ to the domain of the subset $\mathcal{P}^{(1)}$. For any $\boldsymbol{P}_{i}^{\left(2\right)} \in \mathcal{P}^{(2)}$, the map $\Psi\big(\boldsymbol{P}_{i}^{(2)}\big)$ is given by
\begin{align}
\label{eq:Psi}
\Psi\big(\boldsymbol{P}_{i}^{(2)}\big)=\text{Exp}_{\overline{\boldsymbol{P}}^{(1)}}\bigg(\Gamma_{\overline{\boldsymbol{P}}^{(2)}\to\overline{\boldsymbol{P}}^{(1)}}\Big(\text{Log}_{\overline{\boldsymbol{P}}^{(2)}}\big(\boldsymbol{P}_{i}^{(2)}\big)\Big)\bigg)\,.
\end{align}
To enhance the geometric insight, we explicitly describe the three steps comprising the construction of $\Psi$:
\begin{enumerate} \label{enu:FullPsi}

\item 
Project the \ac{SPD} matrix $\boldsymbol{P}_{i}^{(2)}$ to the tangent plane $\mathcal{T}_{\overline{\boldsymbol{P}}^{(2)}}\mathcal{M}$ by $\boldsymbol{S}_{i}^{(2)}=\text{Log}_{\overline{\boldsymbol{P}}^{(2)}}\big(\boldsymbol{P}_{i}^{(2)}\big)$.

\item 
Parallel transport $\boldsymbol{S}_{i}^{(2)}$ from  $\overline{\boldsymbol{P}}^{(2)}$ to $\overline{\boldsymbol{P}}^{(1)}$ by computing $\boldsymbol{S}_{i}^{(2)\to(1)}=\Gamma_{\overline{\boldsymbol{P}}^{(2)}\to\overline{\boldsymbol{P}}^{(1)}}\big(\boldsymbol{S}_{i}^{(2)}\big)$.

\item 
Project the symmetric matrix  $\boldsymbol{S}_{i}^{(2)\to(1)}\!\in\!\mathcal{T}_{\overline{\boldsymbol{P}}^{(1)}}\mathcal{M}$ back to the manifold using $\text{Exp}_{\overline{\boldsymbol{P}}^{(1)}}\big(\boldsymbol{S}_{i}^{(2)\to(1)}\big)$.
\end{enumerate} 


The implementation of $\Psi$ can be simplified and made more efficient by using the following theorem.

\begin{theorem}
\label{thm:EPE}
Let $\boldsymbol{A},\boldsymbol{B},\boldsymbol{P} \!\in\! \mathcal{M}$ and let $\boldsymbol{S}\!=\!\text{Log}_{\boldsymbol{B}}(\boldsymbol{P})\! \in\! \mathcal{T}_{\boldsymbol{B}}\mathcal{M}$. Then,
\[
	\text{Exp}_{\boldsymbol{A}}\left(\Gamma_{{\boldsymbol{B}}\to {\boldsymbol{A}}}\left(\boldsymbol{S}\right)\right)=\boldsymbol{E}\boldsymbol{P}\boldsymbol{E}^{T},
\]
where $\boldsymbol{E}=\big(\boldsymbol{A}\boldsymbol{B}^{-1}\big)^{\tfrac{1}{2}}$.
\end{theorem}

In words, the ``parallel transport'' of an \ac{SPD} matrix $\boldsymbol{P} \in \mathcal{M}$ from $\boldsymbol{B}$ to $\boldsymbol{A}$ is given the same transformation applied to $\boldsymbol{S}=\text{Log}_{\boldsymbol{B}}(\boldsymbol{P})$.
Namely, the ``parallel transport'' of the \ac{SPD} matrix $\boldsymbol{P}$ from $\boldsymbol{B}$ to $\boldsymbol{A}$ is equal to projecting $\boldsymbol{P}$ to the tangent plane at $\boldsymbol{B}$, parallel transporting the projection to the tangent plane at $\boldsymbol{A}$, and then projecting back to the \ac{SPD} manifold. 
As a consequence, we show in the sequel that the map $\Psi$ in \eqref{eq:Psi} can be written simply in terms of $\Gamma$.
The proof of Theorem \ref{thm:EPE} is given in Appendix \ref{sec:proofThm_EPE}. We note that we present the theorem in a general context, since we did not find such a result in the literature and believe it might be of independent interest.

Theorem \ref{thm:EPE} enables us to efficiently compute $\Psi\big(\boldsymbol{P}_{i}^{(2)}\big)$, since it circumvents the computation of the Logarithm and Exponential maps of the \ac{SPD} matrix in steps 1 and 3 above. Instead, the transformation defined by $\boldsymbol{E}$ is computed only once for the entire set, and \eqref{eq:Psi} can be recast as:
\begin{equation}\label{eq:epe}
\Psi\left(\boldsymbol{P}_{i}^{\left(2\right)}\right)=\Gamma_{\overline{\boldsymbol{P}}^{\left(2\right)}\to\overline{\boldsymbol{P}}^{\left(1\right)}}\left(\boldsymbol{P}_{i}^{\left(2\right)}\right)=\boldsymbol{E}\boldsymbol{P}_{i}^{(2)}\boldsymbol{E}^{T}
\end{equation}
where $\boldsymbol{E}\triangleq\Big(\overline{\boldsymbol{P}}^{(1)}\big(\overline{\boldsymbol{P}}^{(2)}\big)^{-1}\Big)^{\tfrac{1}{2}}$. Note that this equality is well defined since any tangent plane to the \ac{SPD} manifold $\mathcal{M}$ is the entire space of symmetric matrices \cite{ferreira2006newton}.

Thus far in the exposition, only the uniqueness of the geodesic curve on the manifold of SPD matrices was exploited, such that the \ac{PT} along the geodesic admits the property in \eqref{eq:3prop}, namely: $\Gamma\left({\varphi}'(0)\right)={\varphi}'(1)$. Importantly, \ac{PT} specifically along the unique geodesic curve exhibits important invariance to the ``relative'' location on the manifold.

\begin{definition}[Equivalent Pairs]
\label{def:eq_pairs}
Two pairs $(\boldsymbol{A}_1,\boldsymbol{B}_1)$ and $(\boldsymbol{A}_2, \boldsymbol{B}_2)$, such that $\boldsymbol{A}_1,\boldsymbol{B}_1,\boldsymbol{A}_2,\boldsymbol{B}_2 \! \in \! \mathcal{M}$, are \emph{equivalent} if there exists an invertible matrix $\boldsymbol{E}$ such that
\begin{align*}
	\boldsymbol{A}_{2}=\Gamma(\boldsymbol{A}_1)=\boldsymbol{E}\boldsymbol{A}_1\boldsymbol{E}^T \\
	\boldsymbol{B}_{2}=\Gamma(\boldsymbol{B}_1)=\boldsymbol{E}\boldsymbol{B}_1\boldsymbol{E}^T
\end{align*}
We denote this relation by
\[
	(\boldsymbol{A}_1,\boldsymbol{B}_1) \sim (\boldsymbol{A}_2, \boldsymbol{B}_2)
\]
\end{definition}

\begin{lemma}
The relation $\sim$ is an equivalence relation.
\end{lemma}
The proof is straight-forward as we show in the following.
\begin{itemize}
\item \emph{Reflexivity} is satisfied by setting $\boldsymbol{E}$ to be the identity matrix.
\item \emph{Symmetry}: if  $\boldsymbol{A}_{2}=\boldsymbol{E}\boldsymbol{A}_1\boldsymbol{E}^T $ then $\boldsymbol{A}_{1}=\boldsymbol{E}^{-1}\boldsymbol{A}_{2}\boldsymbol{E}^{-T}$ and analogously for $\boldsymbol{B}_1, \boldsymbol{B}_2$.
\item \emph{Transitivity}: if $\boldsymbol{A}_{2}=\boldsymbol{E}_{1}\boldsymbol{A}_{1}\boldsymbol{E}_{1}^{T}$ and $\boldsymbol{A}_{3}=\boldsymbol{E}_{2}\boldsymbol{A}_{2}\boldsymbol{E}_{2}^{T}$ then $\boldsymbol{A}_{3}=\boldsymbol{E}\boldsymbol{A}_{1}\boldsymbol{E}^{T}$ where $\boldsymbol{E}=\boldsymbol{E}_{2}\boldsymbol{E}_{1}$ and analogously for $\boldsymbol{B}_1, \boldsymbol{B}_2$.
\end{itemize}

In other words, two pairs are equivalent if the relation of the two matrices in the pair  is given by the same transformation $\Gamma$.
We interpret such equivalent pairs as matrices with equivalent intra-relations (e.g., if $(\boldsymbol{A}_1,\boldsymbol{B}_1) \!\sim\! (\boldsymbol{A}_2, \boldsymbol{B}_2)$, then $d_R(\boldsymbol{A}_1,\boldsymbol{B}_1)\!=\!d_R(\boldsymbol{A}_2,\boldsymbol{B}_2)$), but with a different global position on the manifold. For example, each two pairs in Figure \ref{fig:2X2_Cone_Example} are equivalent pairs.

\begin{prop}
\label{prop:Invariant}
Let $(\boldsymbol{A}_{1},\boldsymbol{B}_{1})$ be a pair of \ac{SPD} matrices $\boldsymbol{A}_{1},\boldsymbol{B}_{1}\!\in\!\mathcal{M}$, and let $\big[(\boldsymbol{A}_{1},\boldsymbol{B}_{1})\big]$ denote the equivalence class
\[
\big[(\boldsymbol{A}_{1},\boldsymbol{B}_{1})\big] = \big\{ (\boldsymbol{A}_2, \boldsymbol{B}_2) \!\in\! \mathcal{M}\times\mathcal{M} \, | \,(\boldsymbol{A}_2,\boldsymbol{B}_2) \!\sim\!(\boldsymbol{A}_1,\boldsymbol{B}_1) \big\},
\]
of all matrix pairs that are equivalent to $(\boldsymbol{A}_{1},\boldsymbol{B}_{1})$.
Then, for any $(\boldsymbol{A}_{2},\boldsymbol{B}_{2}) \!\in\! \big[(\boldsymbol{A}_{1},\boldsymbol{B}_{1})\big]$:
\[
\Gamma\circ\Gamma_{\boldsymbol{B}_{1}\to \boldsymbol{A}_{1}}=\Gamma_{\boldsymbol{B}_{2}\to \boldsymbol{A}_{2}}\circ\Gamma\,,
\]
where $\Gamma\left(\boldsymbol{P}\right)=\boldsymbol{E}\boldsymbol{P}\boldsymbol{E}^{T}$ and $\boldsymbol{E}$ is the transformation defined in Definition \ref{def:eq_pairs}.

\end{prop}
The proof is given in Appendix \ref{sec:proofProp_Invariant}.

An immediate consequence of Proposition \ref{prop:Invariant} is that the domain adaptation via the representation $\Psi$ is invariant to the relative position of $\overline{\boldsymbol{P}}^{(1)}$ and $\overline{\boldsymbol{P}}^{(2)}$ on the manifold, and is constructed equivalently for every pair in the equivalence class $\big[(\overline{\boldsymbol{P}}^{(1)},\overline{\boldsymbol{P}}^{(2)})\big]$.

To demonstrate the importance of the property above, we revisit the illustrating BCI problem. Suppose $(\overline{\boldsymbol{P}}_{A1},\overline{\boldsymbol{P}}_{B1})$ are the Riemannian means of the covariance matrices of Subject A and Subject B recorded in Session 1, and suppose $(\overline{\boldsymbol{P}}_{A2},\overline{\boldsymbol{P}}_{B2})$ are the Riemannian means of the covariance matrices of Subject A and Subject B recorded in Session 2. If $(\overline{\boldsymbol{P}}_{A1},\overline{\boldsymbol{P}}_{B1}) \!\sim\!  (\overline{\boldsymbol{P}}_{A2},\overline{\boldsymbol{P}}_{B2})$, then there exists a transformation $\Gamma$ such that $\Gamma$ encodes the relation between Session 1 and Session 2 whereas the relation of the two subjects is encoded by $\Gamma_{\overline{\boldsymbol{P}}_{B1}\to \overline{\boldsymbol{P}}_{A1}}$ or by $\Gamma_{\overline{\boldsymbol{P}}_{B2}\to \overline{\boldsymbol{P}}_{A2}}$, depending on the session.
Proposition \ref{prop:Invariant} guarantees the consistence of the relation between Subject A and Subject B. Namely, the Riemannian mean of Subject B in Session 2 can be related to the Riemannian mean of Subject A in Session 1 using the relation between the sessions (given by $\Gamma$) and the relation between the two subjects (given either by $\Gamma_{\overline{\boldsymbol{P}}_{B1}\to \overline{\boldsymbol{P}}_{A1}}$ or by $\Gamma_{\overline{\boldsymbol{P}}_{B2}\to \overline{\boldsymbol{P}}_{A2}}$), independently of the relative location of the means on the manifold.


\subsection{Extension to $K$ subsets}

Overall, by Theorem \ref{thm:EPE}, for a general number of subsets $K\geq 2$,  we can apply \ac{PT} using $\Psi$  \eqref{eq:epe} directly to the \ac{SPD} matrices $\mathcal{P}^{(k)}=\big\{\boldsymbol{P}_{i}^{(k)}\big\}$ without projections to and from the tangent plane. 
Let  $\hat{\boldsymbol{P}}$  denote the Riemannian mean of Riemannian means (centroids) $\big\{ \overline{\boldsymbol{P}}^{(k)}\big\} _{k=1}^{K}$ of the subsets, namely,  $$\hat{\boldsymbol{P}}=\arg\min_{\boldsymbol{P}}\sum_{k=1}^{K}d_{R}^{2}\big(\boldsymbol{P},\overline{\boldsymbol{P}}^{(k)}\big).$$
Each subset $\mathcal{P}^{(k)}$ is then parallel transported from its corresponding centroid $\overline{\boldsymbol{P}}^{(k)}$ to $\hat{\boldsymbol{P}}$. Formally, let $\boldsymbol{\Gamma}_{i}^{(k)}$ denote $\boldsymbol{P}_{i}^{(k)}$ after applying \ac{PT}, which is given by
\[
\boldsymbol{\Gamma}_{i}^{(k)}=\Gamma_{\overline{\boldsymbol{P}}^{(k)}\to\hat{\boldsymbol{P}}}\big(\boldsymbol{P}_{i}^{(k)}\big),\qquad\forall i,\,k,
\]
and let $\tilde{\boldsymbol{S}}_{i}^{(k)}$ be the projection of  $\boldsymbol{\Gamma}_{i}^{(k)}$ to the Euclidean tangent space  \eqref{eq:TangentApproxiamtion}:
$$\tilde{\boldsymbol{S}}_{i}^{(k)}=\log\big(\hat{\boldsymbol{P}}^{-\tfrac{1}{2}}\boldsymbol{\Gamma}_{i}^{(k)}\hat{\boldsymbol{P}}^{-\tfrac{1}{2}}\big).$$
This projection, which is further discussed in \cite{barachant2013classification}, can be interpreted as (i) data whitening by  $\tilde{\boldsymbol{\Gamma}}_{i}^{(k)}\!=\!\hat{\boldsymbol{P}}^{-\tfrac{1}{2}}{\boldsymbol{\Gamma}}_{i}^{(k)}\hat{\boldsymbol{P}}^{-\tfrac{1}{2}}$, and (ii) projection to $\mathcal{T}_{\boldsymbol{I}}\mathcal{M}$ where $\boldsymbol{I}$ is the identity matrix. The projected symmetric matrices (vectors) $\tilde{\boldsymbol{S}}_{i}^{(k)}$ indeed reside in a Euclidean space. 
The proposed algorithm is given in Algorithm \ref{alg:TL_using_PT}.

\begin{algorithm}[t]
\textbf{\uline{Input}}\textbf{: }$\big\{ \boldsymbol{P}_{i}^{(1)}\big\} _{i=1}^{N_{1}},\big\{ \boldsymbol{P}_{i}^{(2)}\big\} _{i=1}^{N_{2}},\dots,\big\{ \boldsymbol{P}_{i}^{(K)}\big\} _{i=1}^{N_{K}}$ where $\boldsymbol{P}_{i}^{(k)}$ is the \ac{SPD} matrix associated with the $i$-th element (e.g., high-dimensional time-series) in the $k$-th subset.

\textbf{\uline{Output}}\textbf{:} $\big\{\tilde{ \boldsymbol{S}}_{i}^{(1)}\big\} _{i=1}^{N_{1}},\big\{ \tilde{\boldsymbol{S}}_{i}^{(2)}\big\} _{i=1}^{N_{2}},\dots,\big\{ \tilde{\boldsymbol{S}}_{i}^{(K)}\big\} _{i=1}^{N_{K}}$ where $\tilde{ \boldsymbol{S}}_{i}^{(k)}$ is the
new representation of $\boldsymbol{P}_{i}^{(k)}$ in a Euclidean space.
\begin{enumerate}
\item \textbf{For} each $k\in\left\{ 1,2,\dots,K\right\} $, compute $\overline{\boldsymbol{P}}^{(k)}$ the Riemannian mean of the subset $\big\{ \boldsymbol{P}_{i}^{(k)}\big\} $.
\item Compute $\hat{\boldsymbol{P}}$, the Riemannian mean of $\big\{ \overline{\boldsymbol{P}}^{(k)}\big\} _{k=1}^K$.
\item \textbf{For} all $k$ and all $i$, apply Parallel Transport using \eqref{eq:PT}:
\[
\boldsymbol{\Gamma}_{i}^{(k)}=\Gamma_{\overline{\boldsymbol{P}}^{(k)}\to\hat{\boldsymbol{P}}}\big(\boldsymbol{P}_{i}^{(k)}\big).
\]
\item \textbf{For} all $k$ and all $i$, project the transported matrix to the tangent space via: $$\tilde{\boldsymbol{S}}_{i}^{(k)}=\log\big(\hat{\boldsymbol{P}}^{-\tfrac{1}{2}}\boldsymbol{\Gamma}_{i}^{(k)}\hat{\boldsymbol{P}}^{-\tfrac{1}{2}}\big).$$
\end{enumerate}
\caption{Domain adaptation using Parallel Transport for \ac{SPD} matrices}
\label{alg:TL_using_PT}
\end{algorithm}

We conclude this section with two remarks. First, since the matrices $\tilde{\boldsymbol{S}}_{i}^{\left(k\right)}$ are symmetric, only their upper (or lower) triangular part with a gain factor of $\sqrt{2}$  applied to all non-diagonal elements could be taken into account. 
Second, alternative choices of $\hat{\boldsymbol{P}}$ could also be used, for example, the identity matrix.
Indeed, recently \cite{zanini2017transfer} proposed to align datasets for transfer learning in a similar context using the identity matrix as $\hat{\boldsymbol{P}}$.
However in \cite{zanini2017transfer}, the alignment appeared as an empirical affine transformation, whereas in this work, we provide the geometrical justification and rigorous mathematical analysis.

Specifically, in the case of two subsets with means $\overline{\boldsymbol{P}}_{A}\!\in\!\mathcal{M}$ and $\overline{\boldsymbol{P}}_{B}\!\in\!\mathcal{M}$, the affine transformation presented in \cite{zanini2017transfer} can be interpreted as two consecutive applications of \ac{PT}: from $\overline{\boldsymbol{P}}_{B}$ to $\boldsymbol{I}$ and then from $\boldsymbol{I}$ to $\overline{\boldsymbol{P
}}_{A}$. 
The arbitrary choice of $\boldsymbol{I}$ as an intermediate point introduces dependence of the algorithm on the global position on the manifold. Indeed, such a procedure, which can be expressed by $\Gamma _{\boldsymbol{I} \to \overline{\boldsymbol{P}}_{A}} \circ \Gamma _{\overline{\boldsymbol{P}}_{B} \to \boldsymbol{I}}$ does not admit the invariance property specified in Proposition \ref{prop:Invariant}.

Interestingly, the method proposed in \cite{zanini2017transfer} coincides with the present work, namely, $\Gamma _{\overline{\boldsymbol{P}}_{B} \to \overline{\boldsymbol{P}}_{A}} =\Gamma _{\boldsymbol{I} \to \overline{\boldsymbol{P}}_{A}} \circ \Gamma _{\overline{\boldsymbol{P}}_{B} \to \boldsymbol{I}}$ when the identity matrix  $\boldsymbol{I}$ is on the geodesic $\varphi$ between $\overline{\boldsymbol{P}}_{B}$ and  $\overline{\boldsymbol{P}}_{A}$. In this case, the matrices $\overline{\boldsymbol{P}}_{A}$ and $\overline{\boldsymbol{P}}_{B}$ commute and they have the same eigenvectors (see Appendix \ref{sec:proofCommute}).
From a data analysis perspective, when $\overline{\boldsymbol{P}}_{A}$ and $\overline{\boldsymbol{P}}_{B}$ are two covariance matrices of two subsets, this implies that the subsets have the same principal components.

We set $\hat{\boldsymbol{P}}$ as the Riemannian mean of the centroids so that the overall transport applied to the covariance matrices is minimal. This choice is motivated by the assumption that transporting accumulates distortions. 
This is a straight-forward generalization of the two subsets case, where the parallel transport is carried out along the shortest path (unique geodesic curve).


\section{Experimental Results}
\label{sec:Results}
In this section we show the results of  Algorithm \ref{alg:TL_using_PT} for both a synthetic example and for real data.

\subsection{Toy Problem}
We generate time series in $\mathbb{R}^2$, so that their covariance matrices are in $\mathbb{R}^{2\times 2}$. Since the covariance matrices are symmetric, this particular choice enables us to visualize them in $\mathbb{R}^3$. Concretely, any $2\!\times\! 2$ symmetric matrix $\boldsymbol{A}=\big(\begin{smallmatrix}x & y\\
y & z
\end{smallmatrix}\big)$ can be visualized in $\mathbb{R}^3$ using $(x,y,z)\in \mathbb{R}^3$. $\boldsymbol{A}$  is positive-definite if and only if: $x,z>0$ and $y^2<x z$. These conditions establish the cone manifold of $2\!\times\! 2$ \ac{SPD} matrices.

Consider the set of hidden multi-dimensional times series $\left\{ \boldsymbol{s}_{i}[n]\right\} _{i=1}^{100}$, given by:
\[
\boldsymbol{s}_i\left[n\right]=\begin{bmatrix}\sin\left(2\pi f_{0}n/T\right)\\
\cos\left(2\pi f_{0}n/T+\phi_i\right)
\end{bmatrix},\qquad n=0,\ldots,T-1
\]
where $f_{0}\!=\!10$, $T\!=\!500$, and  $\phi_i$ is uniformly drawn from $\left[-{\pi}\slash{2},0\right]$.
Namely, each time-series $\boldsymbol{s}_i[n]$ consists of two oscillatory signals and is governed by a $1$-dimensional hidden variable $\phi_i$, the initial phase of the oscillations.
Indeed, the population covariance of $\boldsymbol{s}_i[n]$ is
\[
\frac{1}{2}\begin{bmatrix}1 & -\sin\left(\phi_i\right)\\
-\sin\left(\phi_i\right) & 1
\end{bmatrix}
\]
which depends only on $\phi_i$, and therefore, when presenting the population covariances of the time-series $\boldsymbol{s}_i[n]$ in $\mathbb{R}^3$, two coordinates are fixed and only one varies with $i$.

We generate two observable subsets, $\mathcal{X}^{(k)} \!=\! \big\{ \boldsymbol{x}_{i}^{(k)}[n]\big\} _{i=1}^{100}$, $k=1,2$ such that: $$\boldsymbol{x}_{i}^{(k)}\left[n\right]=\boldsymbol{M}^{(k)}\boldsymbol{s}_{i}\left[n\right], 
$$
where $\boldsymbol{M}^{(1)}$ is randomly chosen, and $\boldsymbol{M}^{(2)}\!=\!1.5\big(\begin{smallmatrix}-1 & 0\\
0 & 1
\end{smallmatrix}\big)\boldsymbol{M}^{(1)}$.
The two subsets  $\mathcal{X}^{(1)}$ and $\mathcal{X}^{(2)}$ can be viewed as two different observations of $\big\{ \boldsymbol{s}_{i}[n]\big\}$ through two unknown observation functions $\boldsymbol{M}^{(1)}$ and $\boldsymbol{M}^{(2)}$. For example, $\mathcal{X}^{(1)}$ and $\mathcal{X}^{(2)}$  can represent two different batches, and $\boldsymbol{M}^{(1)}$ and $\boldsymbol{M}^{(2)}$ can represent the discrepancy between two different sessions of a particular experiment.
For each $\boldsymbol{x}_{i}^{(k)}[n]$, we compute its sample covariance matrix by
\[
\boldsymbol{P}_{i}^{(k)}=\tfrac{1}{T}\sum_{n}\boldsymbol{x}_{i}^{(k)}\left[n\right]\big(\boldsymbol{x}_{i}^{\left(k\right)}\left[n\right]\big)^{T}=\boldsymbol{M}^{(k)}\boldsymbol{P}_{s_{i}}\big(\boldsymbol{M}^{(k)}\big)^{T},
\]where  $\boldsymbol{P}_{s_i}$ denotes the inaccessible sample covariance of $\boldsymbol{s}_i[n]$, which is given by:
\[
\boldsymbol{P}_{s_i}=\tfrac{1}{T}\sum_{n=0}^{T-1}\boldsymbol{s}_i\left[n\right]\left(\boldsymbol{s}_i\left[n\right]\right)^{T}.
\]
Our goal is to obtain a new representation of the observed data both in  $\mathcal{X}^{(1)}$ and $\mathcal{X}^{(2)}$, which circumvents the effect of  $\boldsymbol{M}^{(1)}$ and $\boldsymbol{M}^{(2)}$. 
Moreover, in the new representation, we aspire to associate two observations from possibly different subsets which have a similar initial phase $\phi_i$. 


In Figure \ref{fig:Toy-Example} we plot the $2\!\times\! 2$ \ac{SPD} matrices in $\mathbb{R}^3$, where the black points mark the boundaries of the cone manifold. The red line marks the center of the cone, given by $\alpha \boldsymbol{I}$  for $\alpha\!\in\![0,2]$, and the blue point on the red line indicates the identity matrix $\boldsymbol{I}$, namely, where $\alpha = 1$. 

Figure \ref{fig:Toy-Example}(a) presents the two subsets $\mathcal{P}^{(k)}\!=\!\big\{ \boldsymbol{P}_{i}^{(k)}\big\}$, $k\!=\!1,2$ of accessible sample covariance matrices, colored by $\phi_i$ (left) and by $k$ (right).
We observe that the two subsets $\mathcal{P}^{(1)}$ and $\mathcal{P}^{(2)}$ are completely separated, while each subset has a similar structure governed by the values of $\phi_i$.
We apply Steps (1)-(3) of Algorithm \ref{alg:TL_using_PT} to the subsets $\mathcal{P}^{(k)}$ to obtain $\big\{ \boldsymbol{\Gamma}_{i}^{(k)}\big\} $. 

Figure \ref{fig:Toy-Example}(b) presents $\big\{ \boldsymbol{\Gamma}_{i}^{(k)}\big\}$, colored by $\phi_i$ (left) and $k$ (right). Now we observe that in the new representation, the two subsets are aligned, namely the discrepancy caused by $\boldsymbol{M}^{(1)}$ and $\boldsymbol{M}^{(2)}$ is removed while the intrinsic structure given by $\phi_i$ is preserved. As a result, we can associate covariance matrices from different batches but with similar underlying $\phi_i$ values. Note that this was accomplished by Algorithm \ref{alg:TL_using_PT} in a completely unsupervised manner, without access to the hidden ``labels'' $\phi_i$.

\begin{center}
\begin{figure}[h]
\subfloat[]{\includegraphics[width=1\columnwidth]{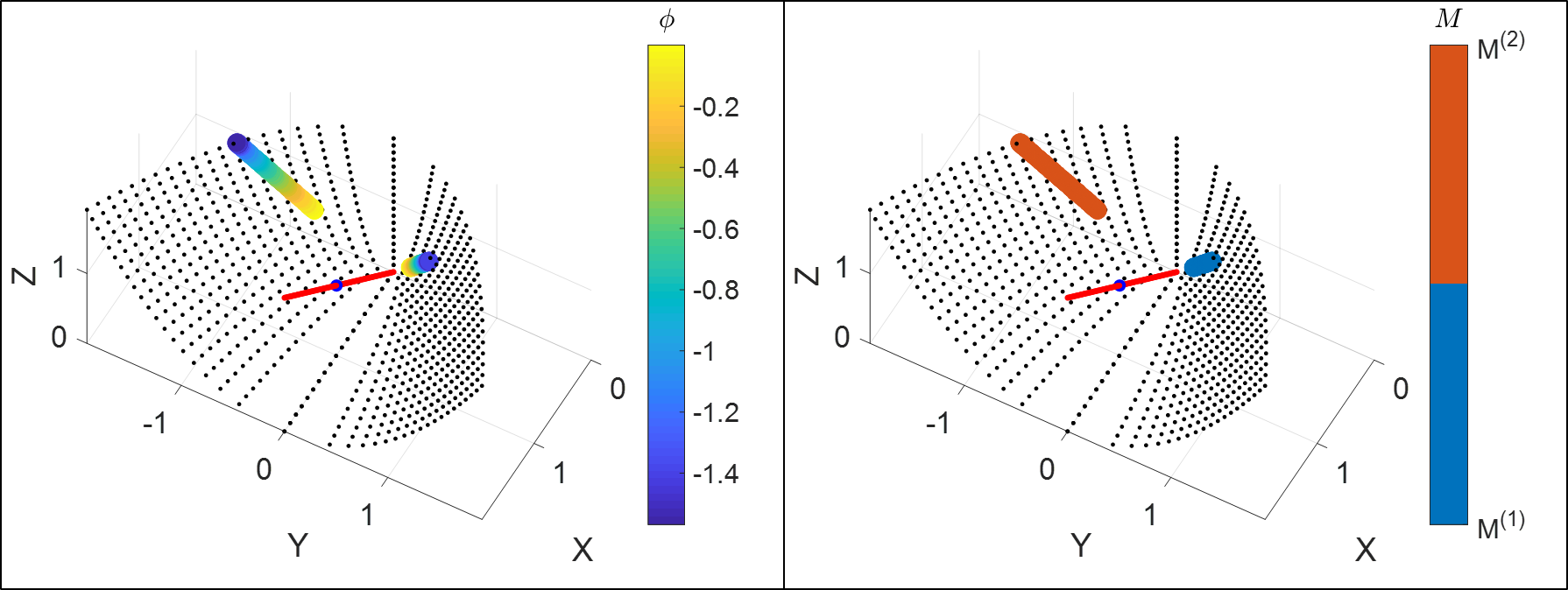}}

\subfloat[]{\includegraphics[width=1\columnwidth]{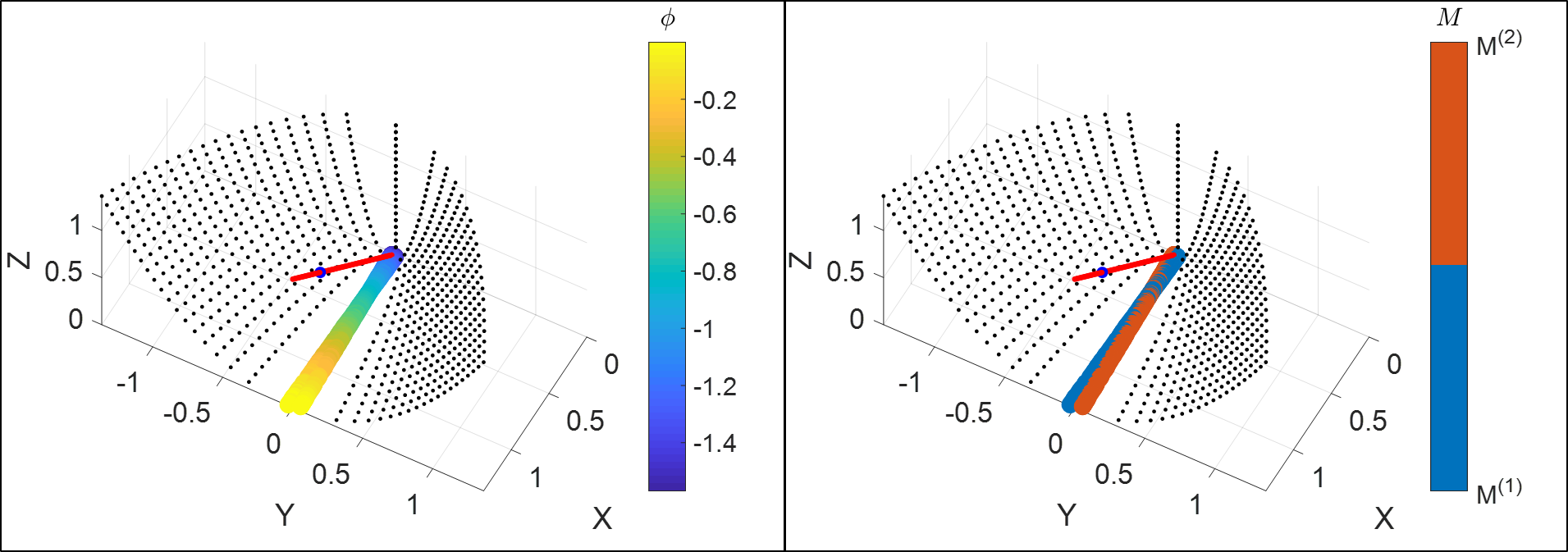}}

\caption{Synthetic example, applying Steps (1)-(3) of Algorithm \ref{alg:TL_using_PT}. (a) Scatter plot of $\mathcal{P}^{(k)}$ colored by $\phi_i$ (left) and by $k$ (right). (b) Scatter plot of $\big\{ {\boldsymbol{\Gamma}}_{i}^{(k)}\big\} $ obtained by Algorithm \ref{alg:TL_using_PT} colored by $\phi_i$ (left) and by $k$ (right).  Note that in the new representation (b), the two subsets are aligned, namely the discrepancy caused by $\boldsymbol{M}^{(1)}$ and $\boldsymbol{M}^{(2)}$ is removed while the intrinsic structure given by $\phi_i$ is preserved.}
\label{fig:Toy-Example}
\end{figure}
\end{center}

\subsection{BCI - Motor Imagery }
\label{sub:BCI}
As described in Section \ref{sec:DomainAdaptaion}, we use data from the BCI competition IV \cite{naeem2006seperability}.
The dataset contains EEG recordings acquired by $22$ EEG electrodes from 9 subjects, where the data from each subject was recorded on 2 different days of experiments. 
The experiment protocol consists of repeated trials, where in each trial the subject was asked to imagine performing one out of four possible movements: (i) right hand, (ii) left hand, (iii) both feet, and (iv) tongue. Overall, in a single day, each movement was repeated $72$ times by each subject, and therefore, the dataset contains $288$ trials from each subject in each day of experiments.

We remark that all the algorithms participating in the competition reported on poor classification results for particular four subjects. 
Since our goal is not to improve the classification of the data from each subject,  we excluded these four subjects (indexed 2,4,5,6).

Initially, focusing on the data from a single subject, we show that Algorithm \ref{alg:TL_using_PT} builds a representation of the data which enables us to train a classifier with data from one day of experiments and apply it to data from the other day of experiments. 
Then, we further show that Algorithm \ref{alg:TL_using_PT} builds a representation that allows us to  train a classifier based on data from one subject and apply it to data from a different subject without any additional labeled trials. 
Finally, we extend the latter result and show the performance on multiple subjects.

\subsubsection{Single Subject -- Different Days}
In the first experiment, we process the recordings of Subject 8 (arbitrarily chosen) from the two days of experiments. We report that the results for the other 4 subjects were similar.
We denote the subsets of trial recordings from day $k=1,2$ by $\mathcal{X}^{(k)}\!=\!\big\{ \boldsymbol{x}_{i}^{(k)}\big\} _{i=1}^{288}$. From the recordings of each trial $i$, we compute the sample covariance matrix $\boldsymbol{P}_{i}^{\left(k\right)}\in\mathbb{R}^{22\times22}$, and denote $\mathcal{P}^{(k)}\!=\!\big\{\boldsymbol{P}_{i}^{(k)}\big\}$. 

To highlight the challenge, we first compute $\hat{\boldsymbol{P}}$, the Riemannian mean of all covariance matrices (from both subsets). Then, we project the matrices onto $\mathcal{T}_{\hat{\boldsymbol{P}}}\mathcal{M}$ by computing $\boldsymbol{B}_{i}^{(k)}\!=\!\text{Log}_{\hat{\boldsymbol{P}}}\big(\boldsymbol{P}_{i}^{(k)}\big)$.  For visualization purpose, we apply tSNE \cite{maaten2008visualizing} to the vectors $\big\{ \boldsymbol{B}_{i}^{(k)}\big\} $. 
Figure \ref{fig:Subject8SingleDay} (a) presents the two dimensional representation of the vectors obtained by the tSNE algorithm. Namely, each point in the figure is the representation of a vector $\boldsymbol{B}_i^{(k)}$. On the left, the points are colored according to the different days (indexed by $k=1,2$), and on the right, the points are colored according to the mental task. We observe that, similarly to the toy problem, the recordings from the different days are completely separated. 
This implies that one cannot train a classifier from the recordings from day 1 and apply it to the recordings from day 2. 

We apply Algorithm \ref{alg:TL_using_PT} to the subsets $\mathcal{P}^{(k)}$ covariance matrices and obtain the subsets $\tilde{\mathcal{S}}^{(k)}\!=\!\big\{ \tilde{\boldsymbol{S}}_{i}^{(k)}\big\} $. Figure\ref{fig:Subject8SingleDay} (b) presents the two dimensional representation of the vectors $\tilde{\mathcal{S}}^{(k)} $ obtained by the tSNE algorithm. On the left, the points are colored according to the different days, and on the right, the points are colored according to the mental task. We observe that the difference between the different days of experiments is completely removed.
More importantly, we further observe that the new representations of the two subsets are aligned, i.e., we obtain similar representations of two recordings associated with the same mental task, regardless of their respective sessions (days of experiments).

\begin{center}
\begin{figure}[h]
\subfloat[]{\includegraphics[width=1\columnwidth]{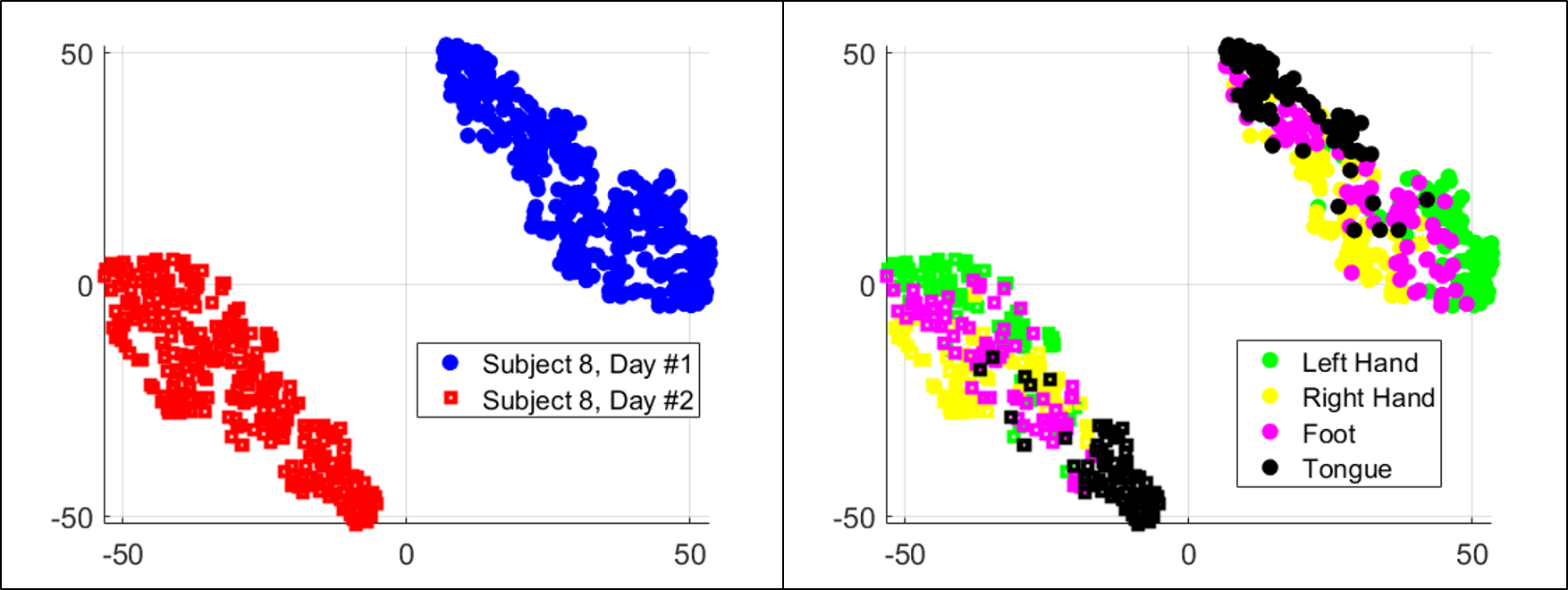}}\\
\subfloat[]{\includegraphics[width=1\columnwidth]{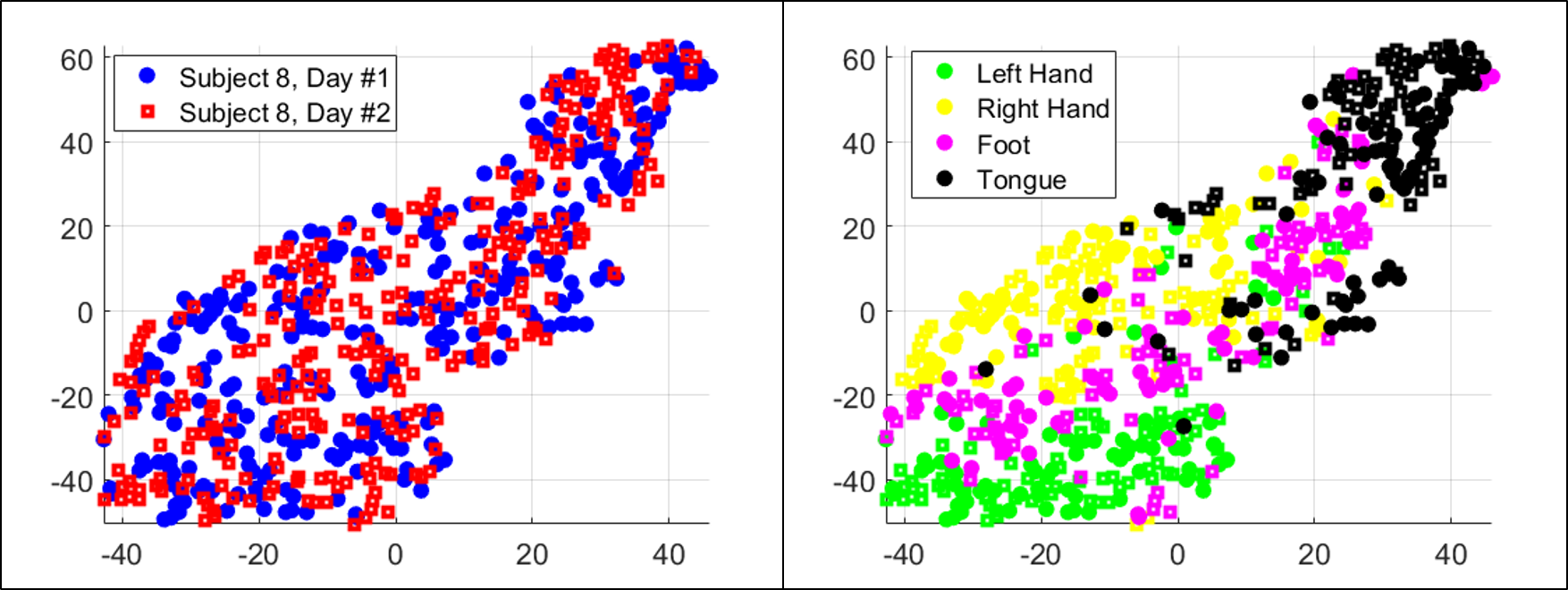}}

\caption{Representation of a single subject's (\#8) recordings from different days. (a) Scatter plot of the ``baseline''  $\big\{ \boldsymbol{B}_{i}^{(k)}\big\} $ colored by the different day (left) and by the mental task (right). (b) Scatter plot of $\big\{ \tilde{\boldsymbol{S}}_{i}^{(k)}\big\} $ obtained by Algorithm \ref{alg:TL_using_PT} colored by the different day (left) and by the mental task (right). Note the difference between the different days of experiments is completely removed. More importantly, we further observe that the new representations of the two subsets are aligned, i.e., we obtain similar representations of two recordings associated with the same mental task, regardless of their respective sessions (days of experiments).}
\label{fig:Subject8SingleDay}
\end{figure}
\end{center}


\subsubsection{Two Subjects}
\label{sub:TwoSubjets}
We repeat the evaluation, but now with the two subsets  $\mathcal{X}^{(k)}$, $k=1,2$ which were recorded from two subjects, specifically, Subject $3$ and Subject $8$. We repeat the steps from the previous examination. We compute $\hat{\boldsymbol{P}}$, the Riemannian mean of all covariance matrices (from both subsets). Then, we project the covariance matrices onto $\mathcal{T}_{\hat{\boldsymbol{P}}}\mathcal{M}$ and apply tSNE to obtain two dimensional representations. Figure \ref{fig:TwoSubjects3and8} (a) presents the two dimensional representations obtained by the tSNE algorithm. On the left, the points are colored by the subject index, and on the right the points are colored according to the mental task. Similarly to the single-subject two-sessions case, we observe that the recordings from different subjects are completely separated in the obtained representation. 

We also apply the mean transport approach presented in \cite{barachant2013classification}. The mean transport is obtained by projecting each subset of covariance matrices $\mathcal{P}^{(k)}$ to its own tangent plane $\mathcal{T}_{\overline{\boldsymbol{P}}^{(k)}}\mathcal{M}$, where $\overline{\boldsymbol{P}}^{(k)}$ is the Riemannian mean of the $k$-th subset. In other words, we compute $\boldsymbol{S}_{i}^{(k)}\!=\!\text{Log}_{\overline{\boldsymbol{P}}^{(k)}}\big(\boldsymbol{P}_{i}^{(k)}\big)$. Figure \ref{fig:TwoSubjects3and8} (b) presents the two dimensional representation obtained by the tSNE algorithm.  We observe that indeed the two subsets are not separated as in Figure \ref{fig:TwoSubjects3and8}(a), however, the inner structure of each subset was not preserved. Thus, this scheme is insufficient and does not support training a classifier based on data from one subject and applying it to data from another subject. 

Finally, we apply Algorithm \ref{alg:TL_using_PT} to the subsets $\mathcal{P}^{(k)}$, and obtain the subsets $\tilde{\mathcal{S}}^{(k)}\!=\!\big\{ \tilde{\boldsymbol{S}}_{i}^{(k)}\big\} $. Figure \ref{fig:TwoSubjects3and8}(c) presents the two dimensional representation of the subsets $\tilde{\mathcal{S}}^{(k)}$ obtained by the tSNE algorithm.  We observe that in this representation, the subsets are not separated. Moreover, the two subsets are aligned according to the mental tasks, and indeed points that correspond to the same mental task assumed a similar value in the new representation. This new representation allows us to train a classifier using recordings from one subject and apply it to recordings from another subject.

\begin{figure}
\subfloat[]{\includegraphics[width=1\columnwidth]{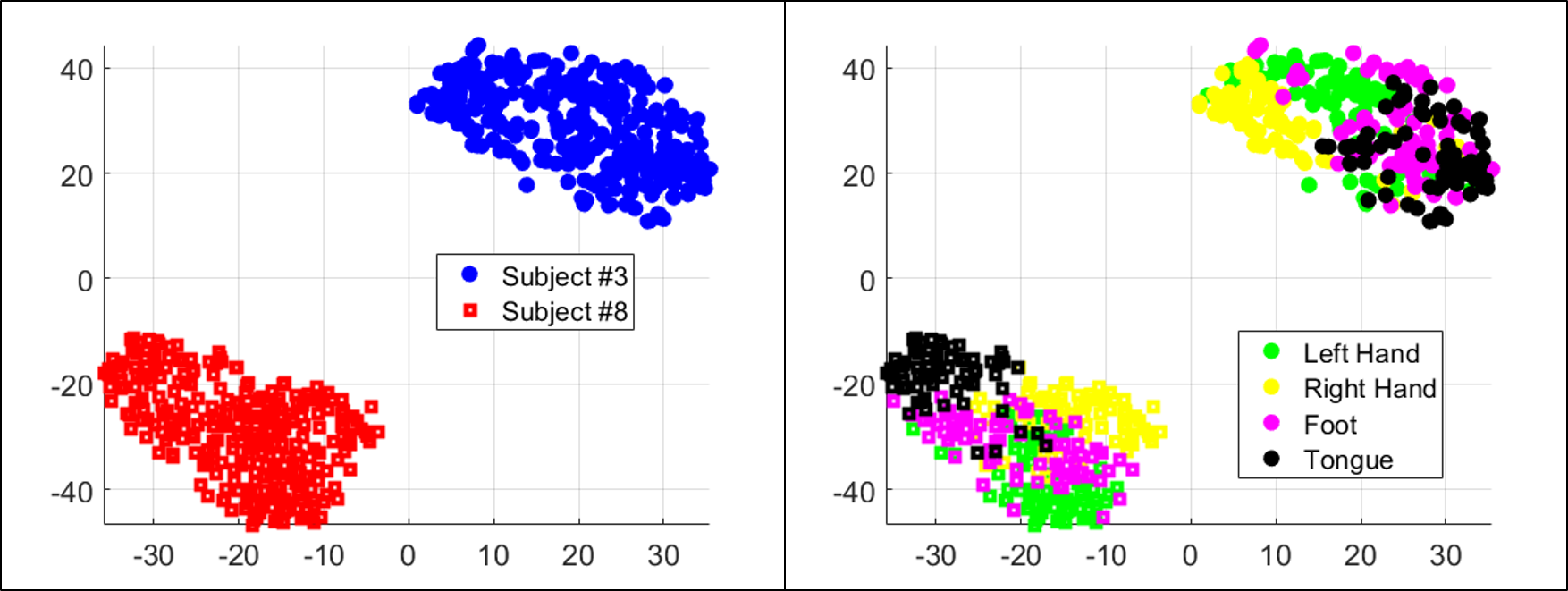}}\\
\subfloat[]{\includegraphics[width=1\columnwidth]{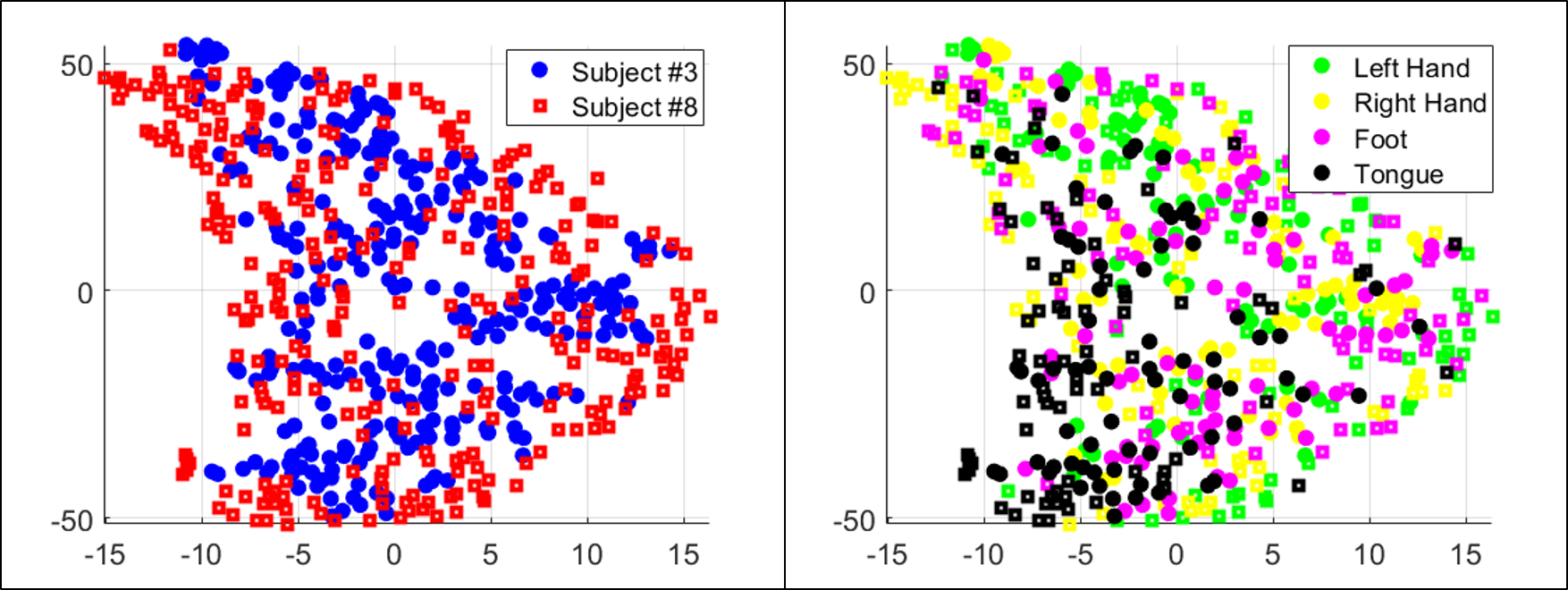}}\\
\subfloat[]{\includegraphics[width=1\columnwidth]{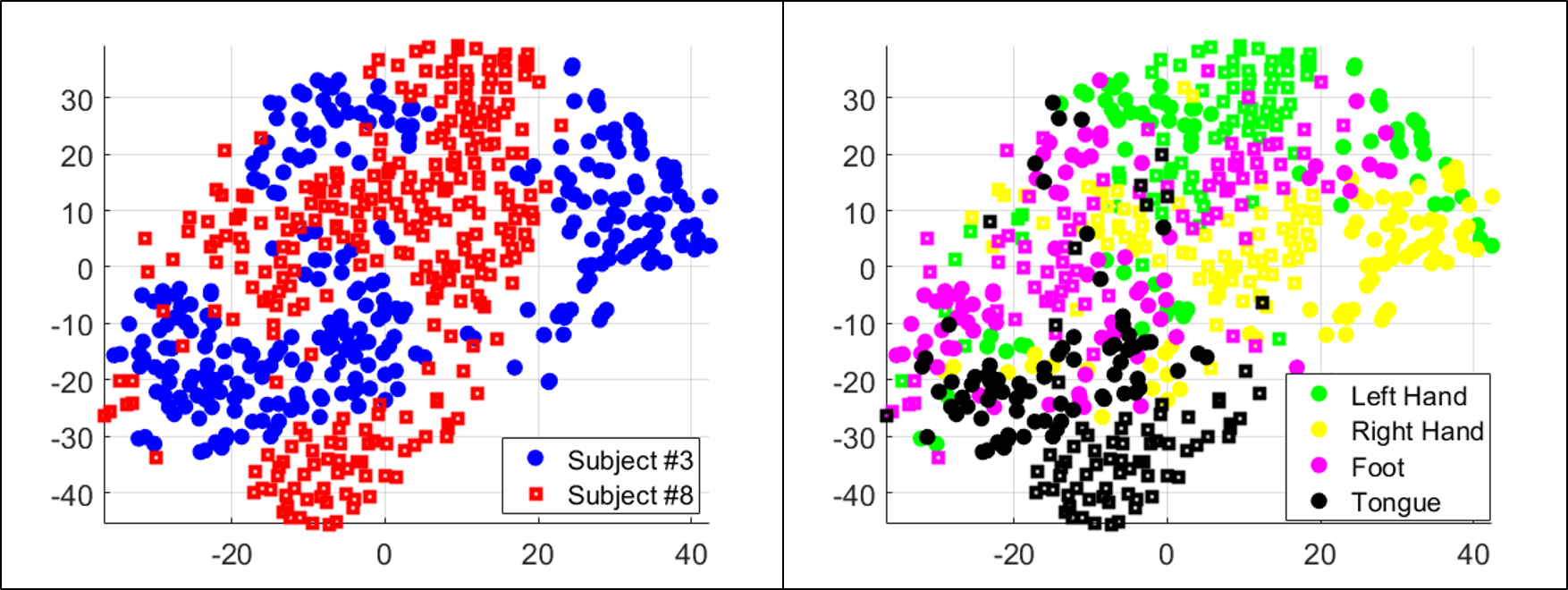}}

\caption{Representation of recordings from two subjects (\#3 and \#8). (a) Scatter plot of the ``baseline'' $\big\{ \boldsymbol{B}_i^{(k)}\big\} $ colored by the subject (left) and by the mental task (right). (b) Scatter plot of the representation of $\big\{ \boldsymbol{S}_{i}^{(k)}\big\}$ obtained by the ``mean transport'', colored by the subject (left) and by the mental task (right). (c) Scatter plot of $\big\{ \tilde{\boldsymbol{S}}_{i}^{(k)}\big\} $ obtained by Algorithm \ref{alg:TL_using_PT} colored by the subject (left) and by the mental task (right). See the text for details.}
\label{fig:TwoSubjects3and8}
\end{figure}

\subsubsection{Multiple Subjects}
\label{sub:MultipleSubjects}
In the third experiment, we apply Algorithm \ref{alg:TL_using_PT} to multiple subjects. We processed data from five subjects (indexed 1,3,7,8 and 9) and from all trials corresponding to only three of the mental tasks: left hand, foot, and tongue. Namely, we omit one class for visualization purposes. We denote the subsets of the recordings by $\mathcal{X}^{(k)}\!=\!\big\{ \boldsymbol{x}_{i}^{(k)}\big\}$ for $k=1,3,7,8,9$. As before, we compute the covariance matrices $\mathcal{P}^{(k)}\!=\!\big\{ \boldsymbol{P}_{i}^{(k)}\big\}$.  We compute $\hat{\boldsymbol{P}}$, the Riemannian mean of all covariance matrices. Then, we project the covariance matrices onto $\mathcal{T}_{\hat{\boldsymbol{P}}}\mathcal{M}$ and apply tSNE to obtain a two dimensional representation. Figure \ref{fig:MultipleSubjects} (a) presents the two dimensional representation obtained by the tSNE algorithm. On the left, the points are colored by the subject index, and on the right the points are colored according to the mental task. As before, we observe that the recordings from different subjects are completely separated. 

Next, we apply Algorithm \ref{alg:TL_using_PT} to the five subsets of covariance matrices $\mathcal{P}^{(k)}$ and obtain five subsets of new representations $\tilde{\mathcal{S}}^{(k)}\!=\!\big\{ \tilde{\boldsymbol{S}}_{i}^{(k)}\big\}$. Figure \ref{fig:TwoSubjects3and8} (b) presents the two dimensional representation of the subsets $\tilde{\mathcal{S}}^{(k)}$ obtained by the tSNE algorithm. We observe that also in the multiple subjects scenario, Algorithm \ref{alg:TL_using_PT} was able to center and align the subsets. 

\begin{figure}[tbh]
\subfloat[]{\includegraphics[width=1\columnwidth]{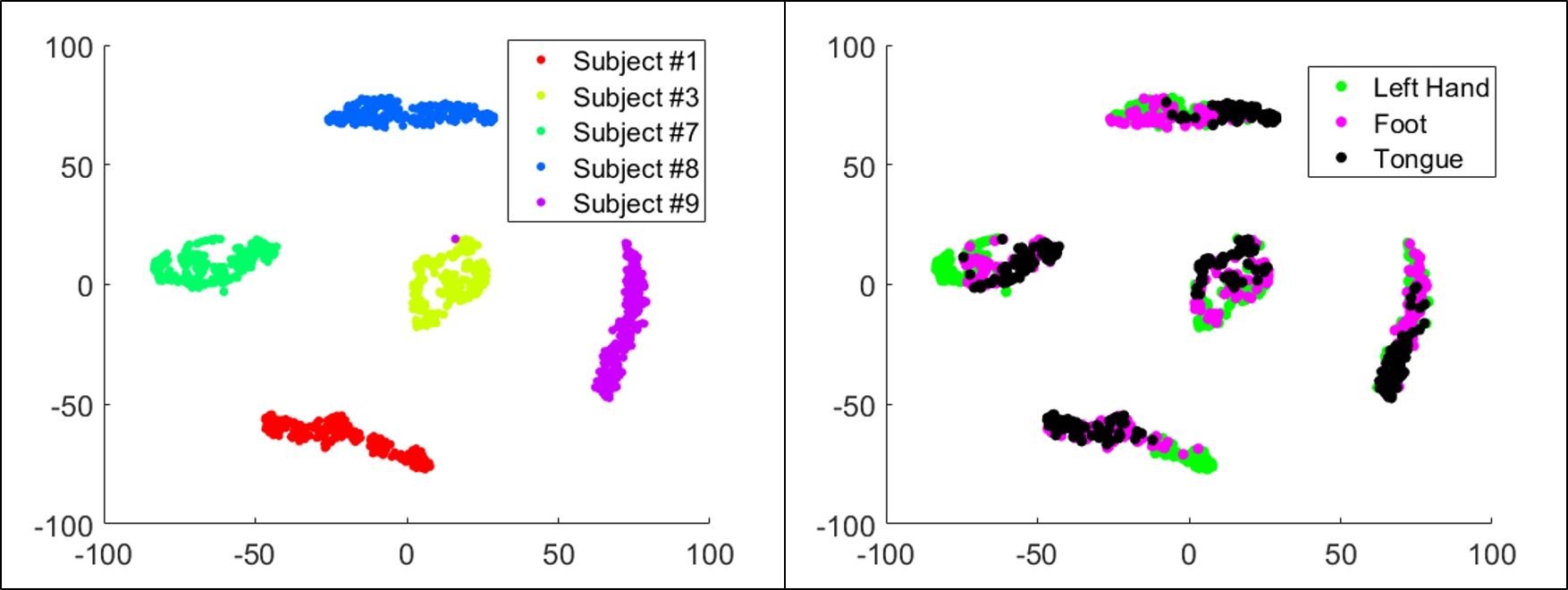}}\\
\subfloat[]{\includegraphics[width=1\columnwidth]{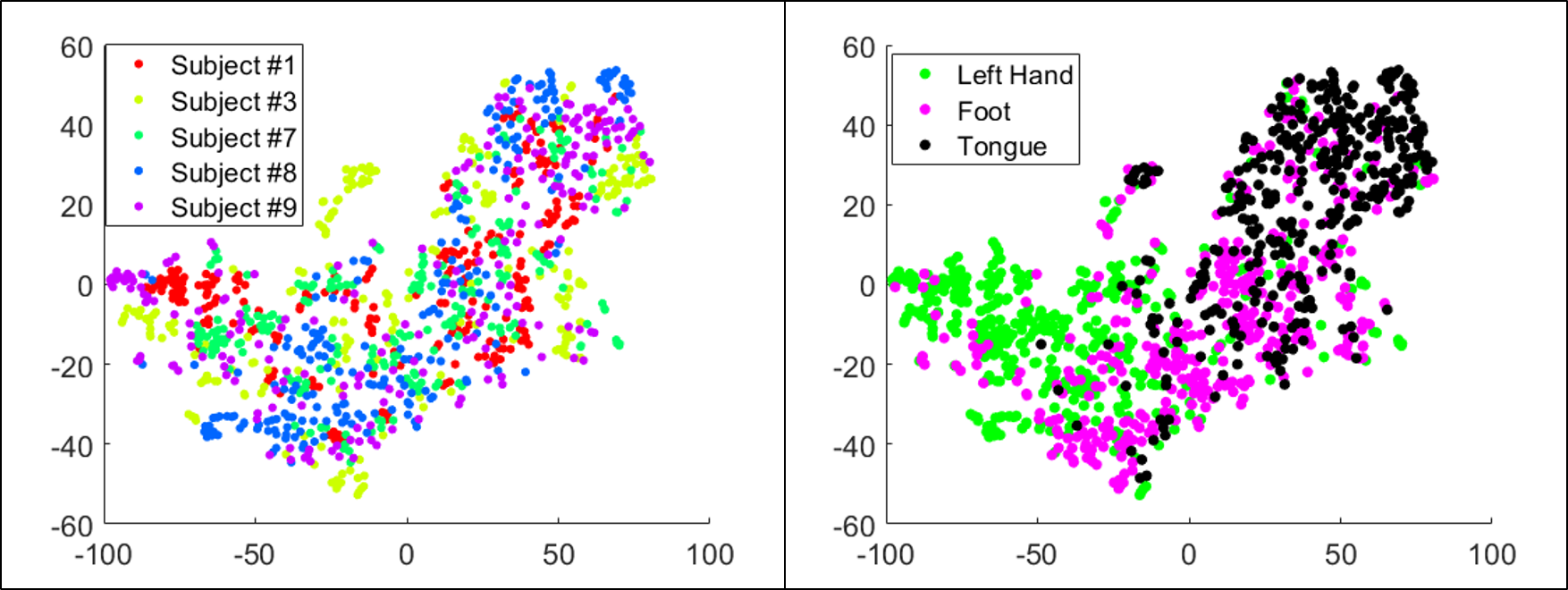}}

\caption{Representation of recordings from multiple subjects. (a) Scatter plot of the ``baseline'' $\big\{ \boldsymbol{B}_{i}^{(k)}\big\} $ colored by the subject (left) and by the mental task (right). (b) Scatter plot of $\big\{ \tilde{\boldsymbol{S}}_{i}^{(k)}\big\} $ obtain by Algorithm \ref{alg:TL_using_PT} colored by the subject (left) and by the mental task (right). See the text for details.}
\label{fig:MultipleSubjects}
\end{figure}

To provide quantitative results, we apply a leave-one-subject-out cross validation, namely, we trained a classifier based on 4 out of the 5 subjects and evaluated the classification accuracy for each one of the three methods mentioned in Section \ref{sub:TwoSubjets}. We compare the classification accuracy of  Algorithm \ref{alg:TL_using_PT}  to the two other approaches denoted by: (i) ``Baseline (No Transport)'', and (ii) the ``Mean Transport'' approach proposed in \cite{barachant2013classification}.

\begin{figure}[b!]
\includegraphics[width=1\columnwidth]{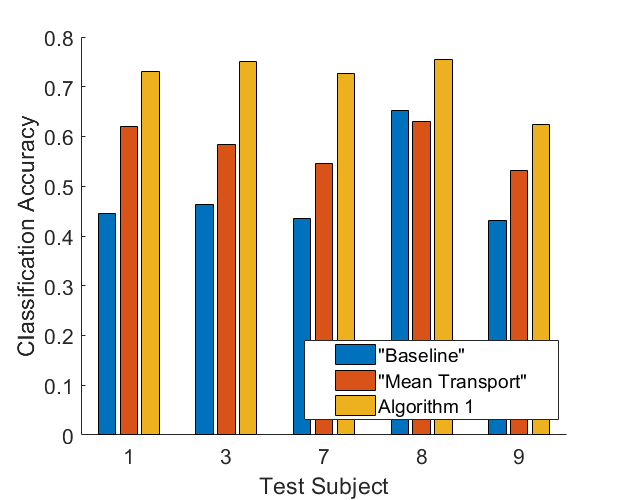}

\caption{The classification accuracy obtained by the three competing methods.}
\label{fig:Bars}
\end{figure}

\begin{figure*}[t!]
\includegraphics[width=0.33\linewidth]{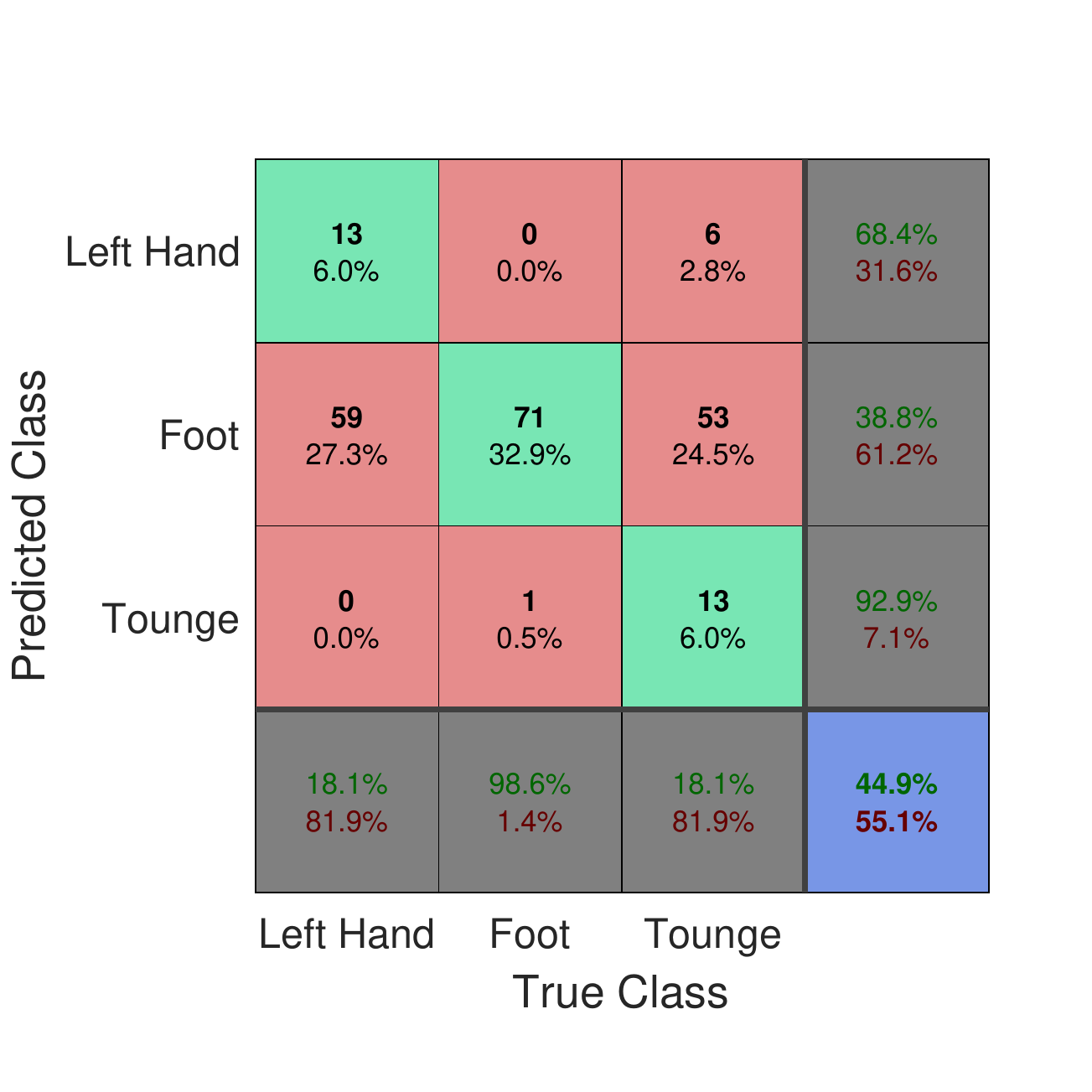}
\includegraphics[width=0.33\linewidth]{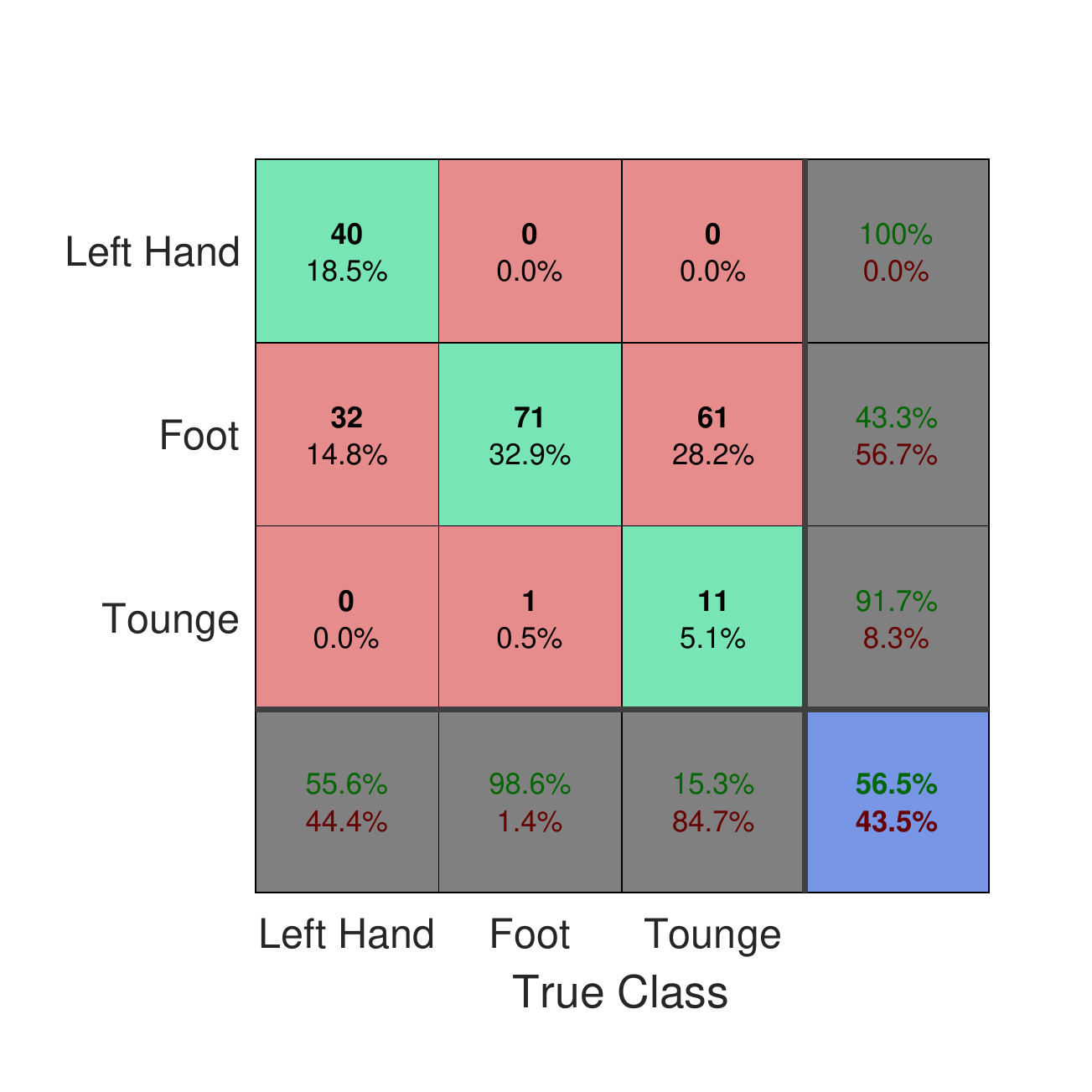}
\includegraphics[width=0.33\linewidth]{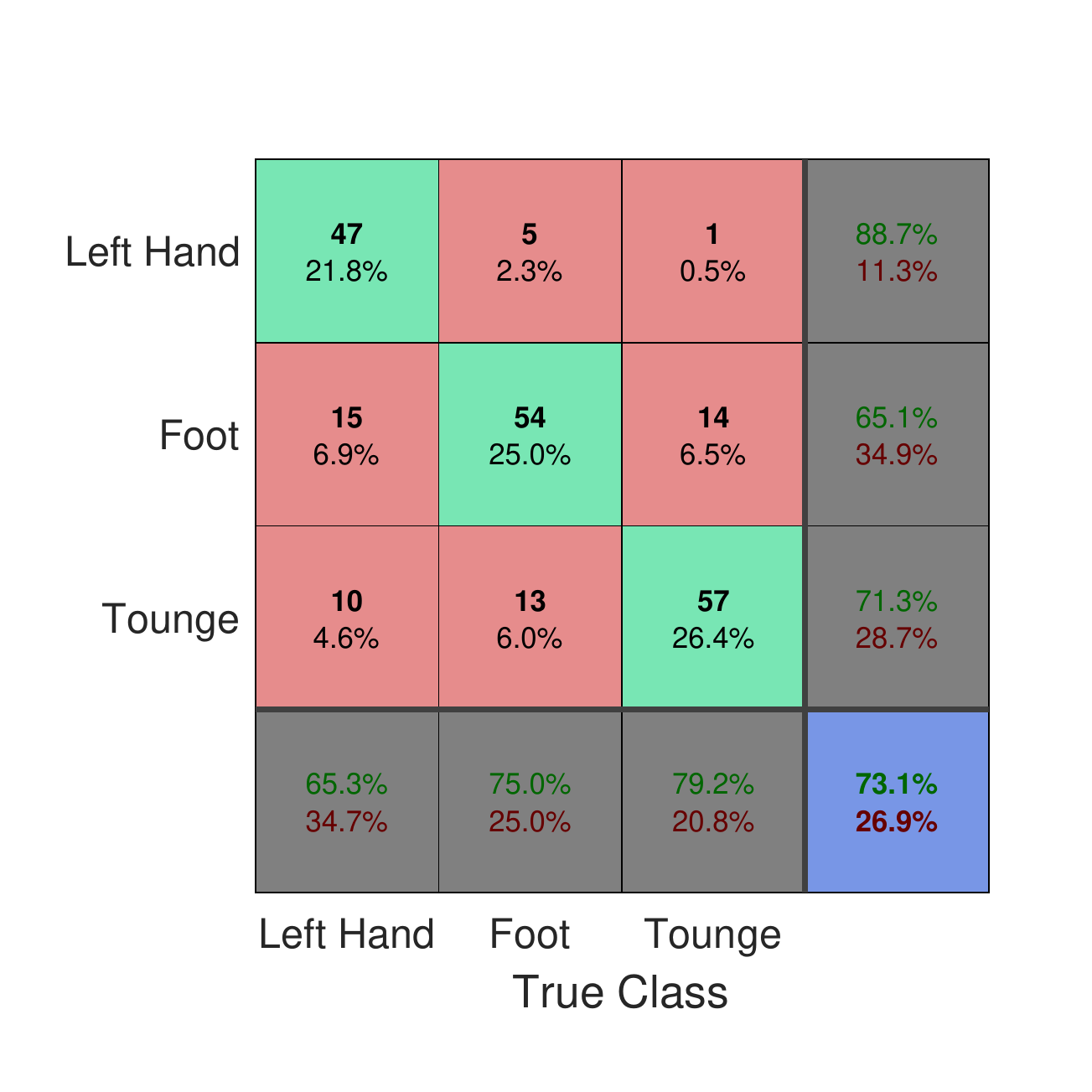}
\caption{The confusion matrices of the BCI task classification of data from Subject 3. (left) Baseline. (center) Mean Transport,  (right) Algorithm \ref{alg:TL_using_PT}.} 
\label{fig:Confusion}
\end{figure*}

Figure \ref{fig:Bars} presents the classification accuracy obtained using the three competing methods as a function of the evaluated subject. In all cases, applying Algorithm \ref{alg:TL_using_PT} to the  data dramatically improved the classification accuracy.
In addition, Figure \ref{fig:Confusion} presents the confusion matrices when Subject 3 was tested. Figure \ref{fig:Confusion} (left) presents the confusion matrix obtained from the data without applying any transportation (``Baseline''). Figure \ref{fig:Confusion} (center) presents the confusion matrix obtained by the ``Mean Transport'' approach. Figure \ref{fig:Confusion} (right) presents the confusion matrix obtained by Algorithm \ref{alg:TL_using_PT}. 
We observe that Algorithm \ref{alg:TL_using_PT} obtains significantly better classification results compared with the ``Baseline'' and ``Mean Transport'' algorithms. In addition, the confusion matrices highlight the challenge in training a classifier from multiple subject data. For example, in both confusion matrices in Figure \ref{fig:Confusion} (left) and (center), the ``foot'' mental task falsely dominated the prediction (it was predicted 183 and 164 times, respectively, whereas is was performed only 72 times). This implies that the decision regions of the classifiers are completely misaligned with the data from a new unseen subject.



\subsection{Sleep Stage Identification}
Here, we demonstrate the applicability of Algorithm \ref{alg:TL_using_PT} to real medical signals. Specifically, we address the problem of sleep stage identification. Typically, for this purpose, data are collected in sleep clinics with multiple multimodal sensors, and then, analyzed by a human expert. There are six different sleep stages: awake, REM, and sleep stages 1-4, indicating shallow to deep sleep.

The data we used is available online in \cite{PhysioNet} and described in detail in \cite{kemp2000analysis}. A single patient's night recording contains several measurements including two EEG channels and one electrooculography (EOG) channel sampled at 100[Hz].
We used recordings from three subjects. We split each subject's night into non-overlapping $30$ seconds windows. We omit the awake and sleep stage 4 windows due to too few occurrences. For visual purposes, we kept only windows corresponding to REM and stage 3.

We denote the $i$-th window of the $k$-th subject by $\boldsymbol{x}_{i}^{(k)}(t)$ with its corresponding covariance matrix $\boldsymbol{P}_i^{(k)}\!\in\! \mathbb{R}^{3\times 3}$. 
We first compute $\hat{\boldsymbol{P}}$, the Riemannian mean of all covariance matrices (from the three subsets). Then, we project the matrices onto $\mathcal{T}_{\hat{\boldsymbol{P}}}\mathcal{M}$ by computing $\boldsymbol{S}_{i}^{(k)}\!=\!\text{Log}_{\hat{\boldsymbol{P}}}\big(\boldsymbol{P}_{i}^{(k)}\big)$. For visualization purposes, we apply PCA to the vectors $\big\{ \boldsymbol{S}_{i}^{(k)}\big\}$ and present the first two principle components\footnotemark.
\footnotetext{Since the covariance matrices in this experiment are of size $3 \times 3$, dimension reduction using PCA was sufficient. It was preferred here over tSNE since it preserves the global geometry of the data.}

\begin{figure}[b!]
\subfloat[]{\includegraphics[width=1\columnwidth]{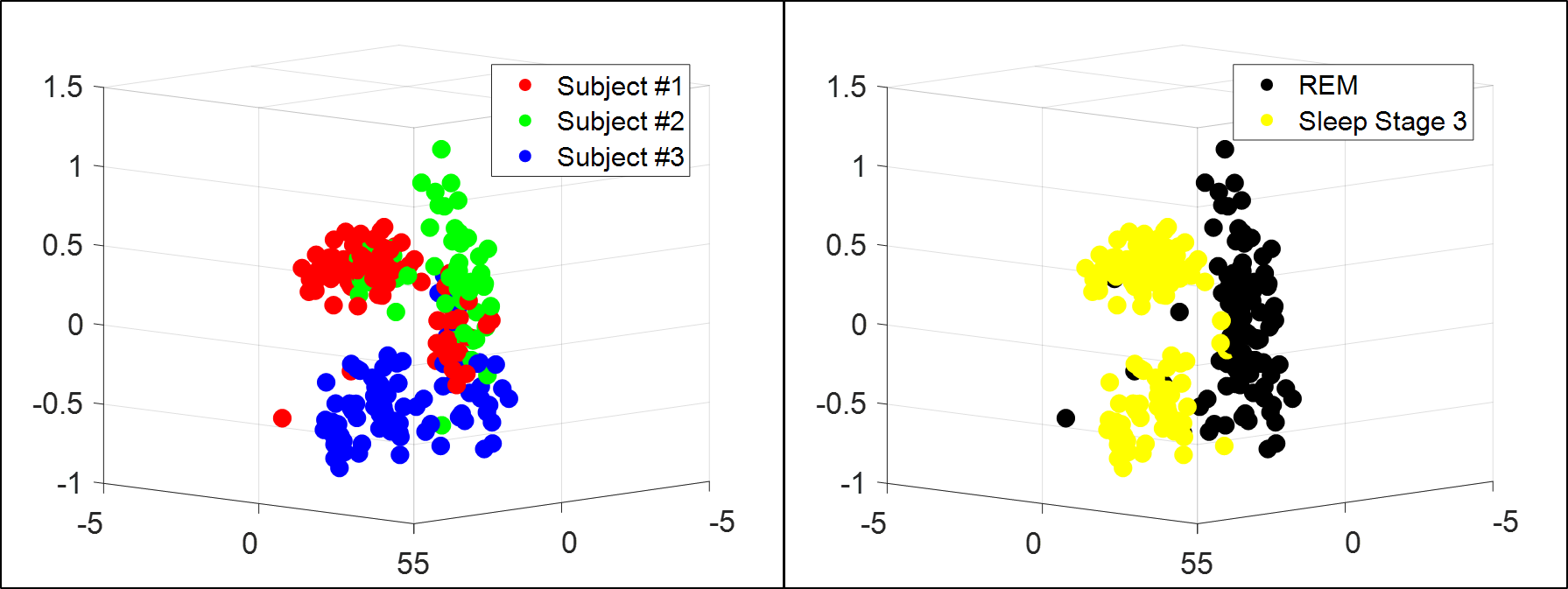}}\\
\subfloat[]{\includegraphics[width=1\columnwidth]{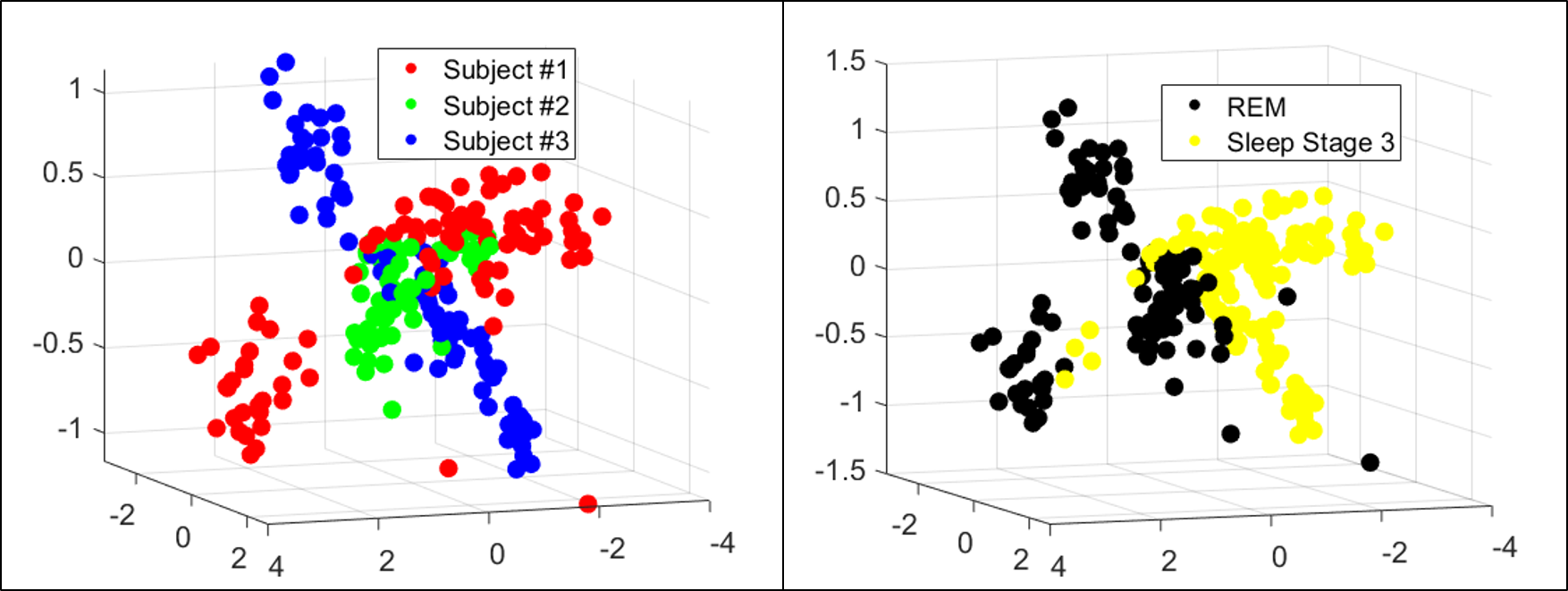}}\\
\subfloat[]{\includegraphics[width=1\columnwidth]{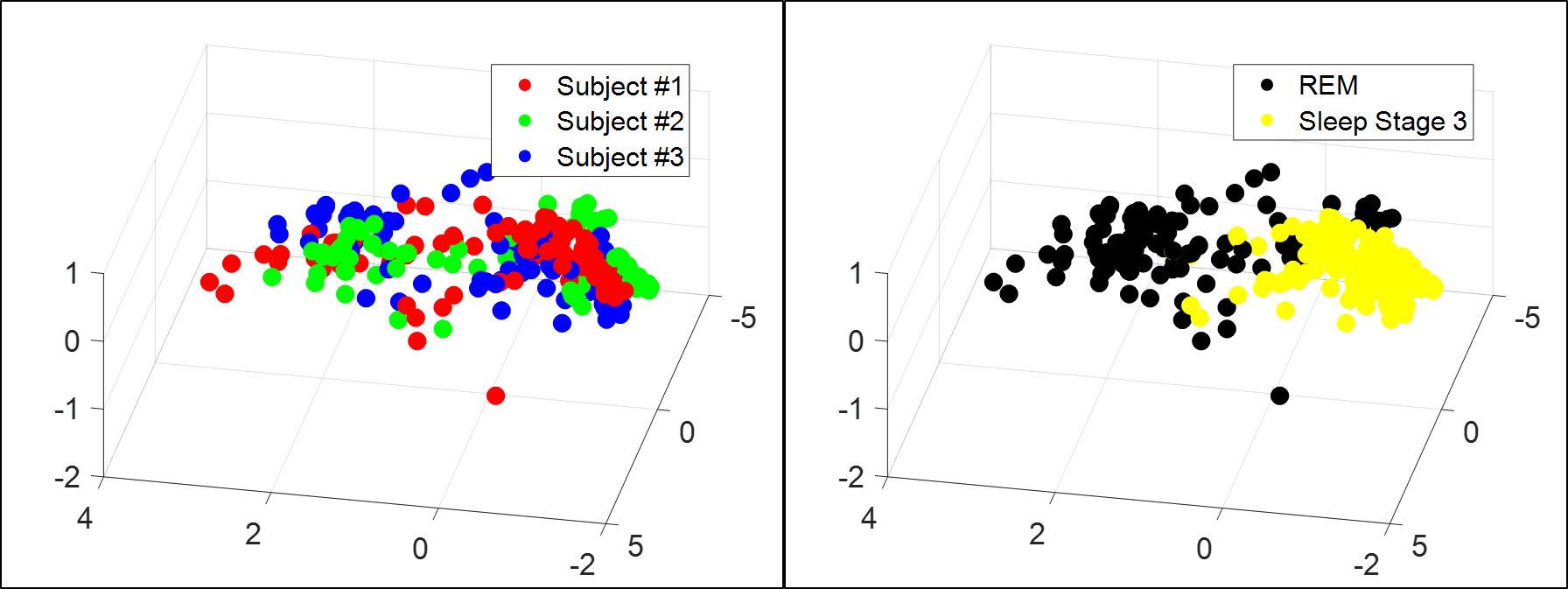}}
\caption{Representation of recordings for sleep stage identification. (a) Scatter plot of the "baseline" $\big\{ \boldsymbol{B}_{i}^{(k)}\big\} $ (after PCA) colored by subject (left) and by sleep stage (right). (b) Scatter plot of $\big\{ \boldsymbol{S}_{i}^{(k)}\big\} $ (after PCA) obtained by the  "mean transport", colored by subject (left) and by sleep stage (right). (c) Scatter plot of $\big\{ \tilde{\boldsymbol{S}}_{i}^{(k)}\big\} $ (after PCA) obtained by Algorithm \ref{alg:TL_using_PT} colored by subject (left) and by sleep stage (right).}
\label{fig:Sleep}
\end{figure}

\begin{figure*}[t!]
\begin{center}
\includegraphics[width=0.324\linewidth]{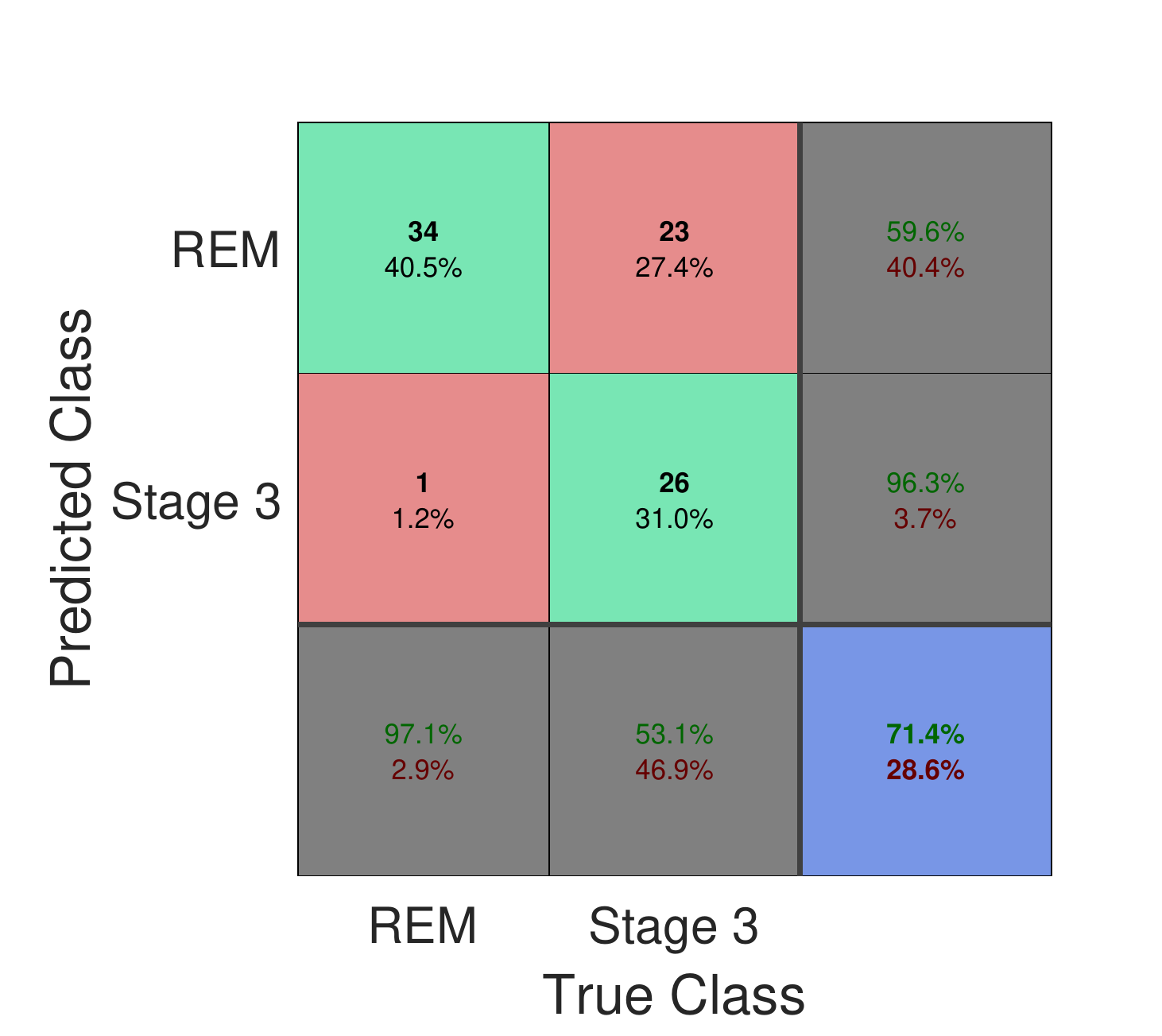}
\includegraphics[width=0.29\linewidth]{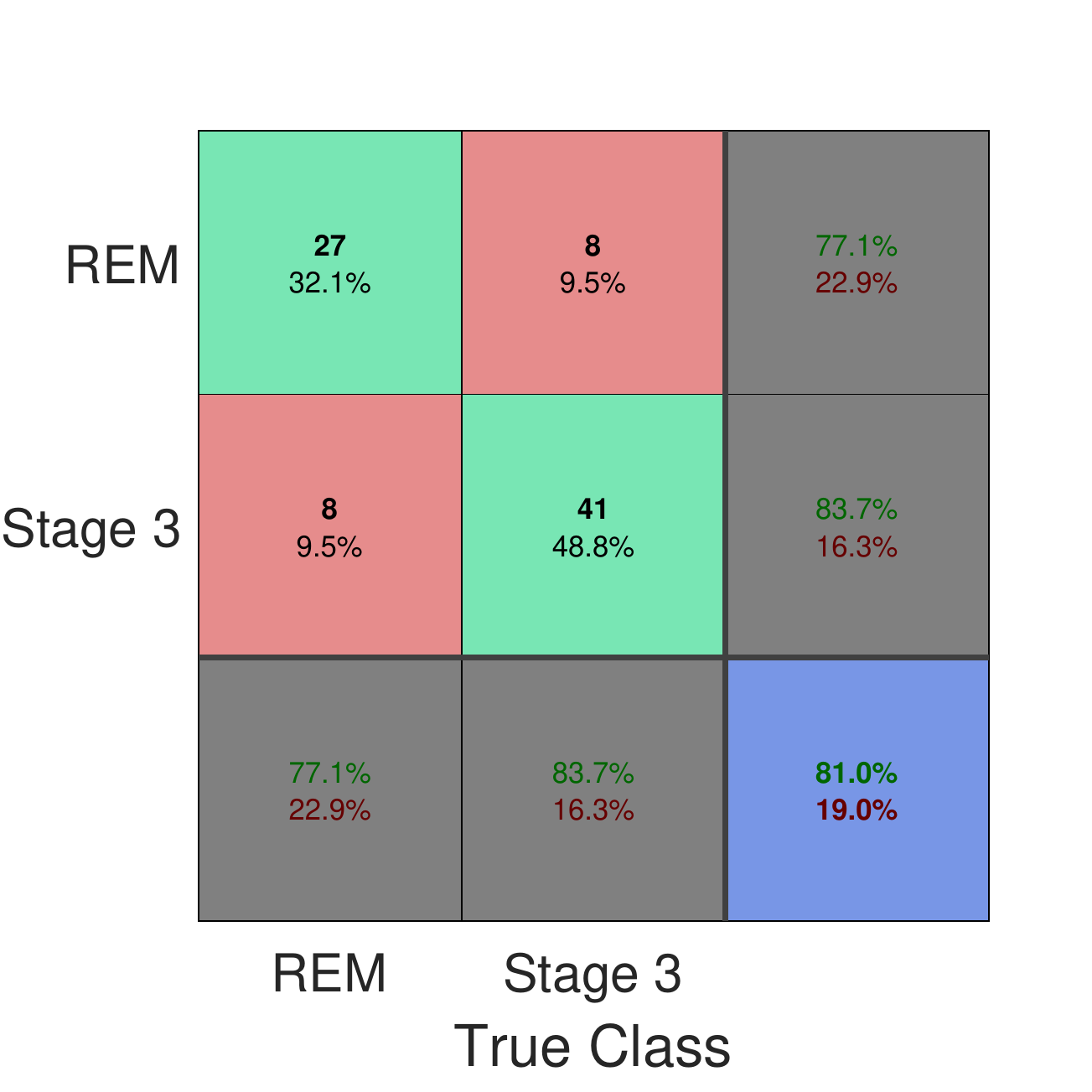}
\includegraphics[width=0.29\linewidth]{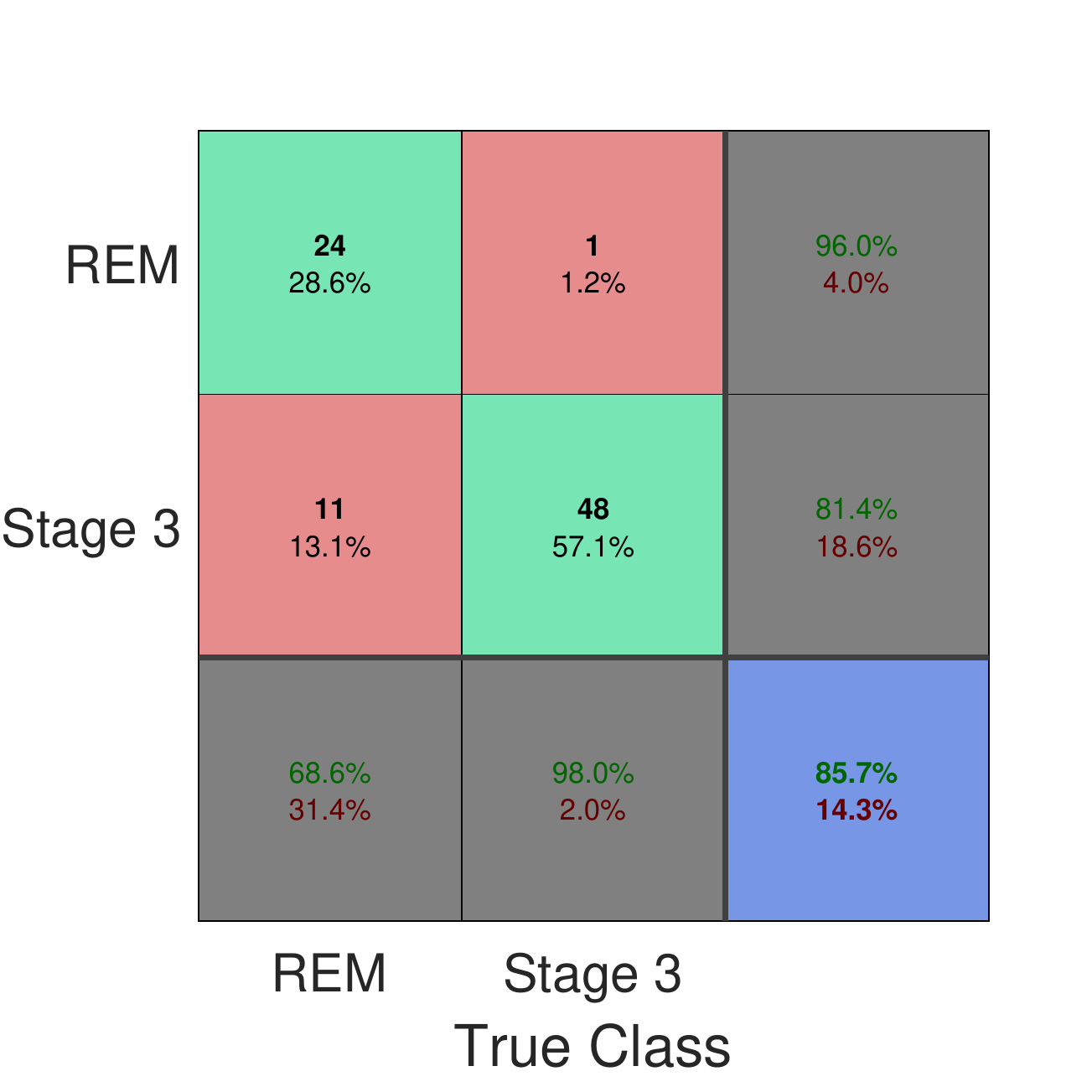}
\caption{Confusion matrices of sleep stage identification based on recordings from Subject 3. (left) Baseline, (center) Mean Transport, (right) Algorithm 1. }
\label{fig:SleepCM}
\end{center}
\end{figure*}

Figure \ref{fig:Sleep} (a) presents the two dimensional representation of the vectors obtained by PCA. Namely, each point in the figure is the representation of a vector $\boldsymbol{S}_i^{(k)}$. On the left, the points are colored according to the different subjects, and on the right, the points are colored according to the sleep stage. We observe that the points are clustered according to the different subjects.
We apply Algorithm \ref{alg:TL_using_PT} to the three subsets $\big\{ \boldsymbol{P}_{i}^{(k)}\big\} $ of covariance matrices and obtain the subsets $\big\{ \tilde{\boldsymbol{S}}_{i}^{(k)}\big\}$. Figure \ref{fig:Sleep} (b) presents the two dimensional representation of the vectors $\big\{ \tilde{\boldsymbol{S}}_{i}^{k)}\big\} $ obtained by PCA. On the left, the points are colored according to the different subjects, and on the right, the points are colored according to the sleep stage. Now we observe that the data is clustered according to the sleep stage while the difference between the three subjects is completely removed.

As in the Subsection \ref{sub:MultipleSubjects}, to provide quantitative results, we train a classifier based on Subject 1 and Subject 2 and evaluate the classification accuracy on Subject 3. 
Figure \ref{fig:SleepCM} presents the obtained confusion matrices. Figure \ref{fig:SleepCM} (left) presents the confusion matrix obtained from the data without applying any adaptation (``Baseline''). Figure \ref{fig:SleepCM} (center) presents the confusion matrix obtained by the ``Mean Transport'' approach. Figure \ref{fig:SleepCM} (right) presents the confusion matrix obtained by Algorithm \ref{alg:TL_using_PT}. 
We observe that using Algorithm \ref{alg:TL_using_PT} demonstrates better classification results compared with the ``Baseline'' and the ``Mean Transport'' algorithms. 



\section{Conclusions}
\label{sec:Conclusions}
Analyzing complex data in high-dimension is challenging, since such data do not live in a Euclidean space. Therefore, basic operations such as comparisons, additions, and subtractions, which are the basis of any analysis and learning technique, do not necessarily exist and are not appropriately  defined. 
In this work, we propose to view the complex data through the lens of \ac{SPD} matrices, which reside on an analytic Riemannian manifold. 
Using the Riemannian geometry of \ac{SPD} matrices, we presented an approach for multi-domain data representation. Based on this new representation, we proposed an algorithm for domain adaptation. 
We extend the existing results in the Riemannian geomery of \ac{SPD} matrices and establish a framework for the justification and analysis of the proposed solution.
We demonstrated the usefulness of the presented domain adaptation method in applications to simulation and real recorded data.

\appendices

\section{Proofs of lemma \ref{def:PT}}
\label{sec:proofLemma_ESE}
We prove lemma \ref{def:PT}  presented in \cite[Eq.~3.4]{sra2015conic}.
\begin{proof}
The PT of ${\boldsymbol{S}}$  along the geodesic between $\boldsymbol{B}$ and $\boldsymbol{A}$ is given by \cite{ferreira2006newton}:
$$\Gamma_{\boldsymbol{B}\rightarrow \boldsymbol{A}}\left(\boldsymbol{S}\right)=\boldsymbol{M}\boldsymbol{S}\boldsymbol{M}^{T}$$
where $\boldsymbol{M}=\boldsymbol{B}^{\frac{1}{2}}\exp\left(\boldsymbol{B}^{-\frac{1}{2}}\frac{1}{2}\text{Log}_{\boldsymbol{B}}\left(\boldsymbol{A}\right)\boldsymbol{B}^{-\frac{1}{2}}\right)\boldsymbol{B}^{-\frac{1}{2}}$. Now we will show that $\boldsymbol{M}$ can be written more simply, proving a more efficient way to compute it.
We have
\begin{align*}
\boldsymbol{M} & =\boldsymbol{B}^{\frac{1}{2}}\exp\left(\boldsymbol{B}^{-\frac{1}{2}}\frac{1}{2}\text{Log}_{\boldsymbol{B}}\left(\boldsymbol{A}\right)\boldsymbol{B}^{-\frac{1}{2}}\right)\boldsymbol{B}^{-\frac{1}{2}}\\
 & =\boldsymbol{B}^{\frac{1}{2}}\exp\left(\frac{1}{2}\log\left(\boldsymbol{B}^{-\frac{1}{2}}\boldsymbol{A}\boldsymbol{B}^{-\frac{1}{2}}\right)\right)\boldsymbol{B}^{-\frac{1}{2}}\\
 & =\boldsymbol{B}^{\frac{1}{2}}\left(\boldsymbol{B}^{-\frac{1}{2}}\boldsymbol{A}\boldsymbol{B}^{-\frac{1}{2}}\right)^{\frac{1}{2}}\boldsymbol{B}^{-\frac{1}{2}}
\end{align*}
and also
$$\boldsymbol{M}^{2}=\left(\boldsymbol{B}^{\frac{1}{2}}\left(\boldsymbol{B}^{-\frac{1}{2}}\boldsymbol{A}\boldsymbol{B}^{-\frac{1}{2}}\right)^{\frac{1}{2}}\boldsymbol{B}^{-\frac{1}{2}}\right)^{2}=\boldsymbol{A}\boldsymbol{B}^{-1}=\boldsymbol{E}^{2}$$
and since $\boldsymbol{A}\boldsymbol{B}^{-1}$ is similar to $\boldsymbol{B}^{-\frac{1}{2}}\boldsymbol{A}\boldsymbol{B}^{-\frac{1}{2}}>0$, it has only positive eigenvalues and the square root is unique, namely $\boldsymbol{E}=\boldsymbol{M}$.
\end{proof}

\section{Proof of theorem \ref{thm:Gamma}}

\label{sec:proofThm_Gamma}
For better readability we denote $\boldsymbol{A}=\overline{\boldsymbol{P}}^{\left(1\right)}$ and $\boldsymbol{B}=\overline{\boldsymbol{P}}^{\left(2\right)}$.
First we remark that $\Gamma_{\boldsymbol{B}\rightarrow \boldsymbol{A}}$ is well defined since $\mathcal{T}_{\boldsymbol{P}}\mathcal{M}$ is the space of all symmetric matrices regardless the matrix $\boldsymbol{P}$\footnotemark, and  if the input $\boldsymbol{S}$ is symmetric then by definition $\Gamma_{\boldsymbol{B}\rightarrow \boldsymbol{A}}\left(\boldsymbol{S}\right)$ is symmetric as well. 
\footnotetext{Any tangent plane to the \ac{SPD} manifold $\mathcal{M}$ is the entire space of symmetric matrices \cite{ferreira2006newton}.}
\begin{proof}[Proof of Theorem \ref{thm:Gamma}]
Condition (1) is immediate since $\Gamma$ is a linear operation, and therefore we have
$$\sum_{i=1}^{N_{2}}\Gamma\left(\boldsymbol{S}_{i}^{\left(2\right)}\right)=\sum_{i=1}^{N_{2}}\boldsymbol{E}\boldsymbol{S}_{i}^{\left(2\right)}\boldsymbol{E}^{T}=\boldsymbol{E}\underbrace{\sum_{i=1}^{N_{2}}\boldsymbol{S}_{i}^{\left(2\right)}}_{=0}\boldsymbol{E}^{T}=0$$
Conditions (2) and (3) are derived from Lemma \ref{def:PT}, since these are properties of \ac{PT}.
For completeness, we provide their explicit proofs.
Proof of condition (2):
Let $\boldsymbol{A},\boldsymbol{B}\in\mathcal{M}$ and $\boldsymbol{S}_{1},\boldsymbol{S}_{2}\in\mathcal{T}_{\boldsymbol{B}}\mathcal{M}$ and denote $\boldsymbol{E}=\left(\boldsymbol{A}\boldsymbol{B}^{-1}\right)^{\frac{1}{2}}=\boldsymbol{B}^{\frac{1}{2}}\left(\boldsymbol{B}^{-\frac{1}{2}}\boldsymbol{A}\boldsymbol{B}^{-\frac{1}{2}}\right)^{\frac{1}{2}}\boldsymbol{B}^{-\frac{1}{2}}$.
We have
\begin{align*}
\boldsymbol{A}^{-1}\boldsymbol{E} & =\boldsymbol{A}^{-1}\boldsymbol{B}^{\frac{1}{2}}\left(\boldsymbol{B}^{-\frac{1}{2}}\boldsymbol{A}\boldsymbol{B}^{-\frac{1}{2}}\right)^{\frac{1}{2}}\boldsymbol{B}^{-\frac{1}{2}}\\
 & =\boldsymbol{A}^{-1}\boldsymbol{B}^{\frac{1}{2}}\boldsymbol{B}^{-\frac{1}{2}}\boldsymbol{A}\boldsymbol{B}^{-\frac{1}{2}}\left(\boldsymbol{B}^{-\frac{1}{2}}\boldsymbol{A}\boldsymbol{B}^{-\frac{1}{2}}\right)^{-\frac{1}{2}}\boldsymbol{B}^{-\frac{1}{2}}\\
 & =\boldsymbol{B}^{-\frac{1}{2}}\left(\boldsymbol{B}^{-\frac{1}{2}}\boldsymbol{A}\boldsymbol{B}^{-\frac{1}{2}}\right)^{-\frac{1}{2}}\boldsymbol{B}^{-\frac{1}{2}}
\end{align*}
Namely,  $\boldsymbol{A}^{-1}\boldsymbol{E}$ is a symmetric matrix.
Thus, we get that
\begin{align*}
\boldsymbol{E}^{T}\boldsymbol{A}^{-1}\boldsymbol{E} & =\boldsymbol{E}^{T}\boldsymbol{E}^{T}\boldsymbol{A}^{-1}\\
 & =\left(\boldsymbol{A}\boldsymbol{B}^{-1}\right)^{T}\boldsymbol{A}^{-1}\\
 & =\left(\boldsymbol{B}^{-1}\boldsymbol{A}\right)\boldsymbol{A}^{-1}\\
 & =\boldsymbol{B}^{-1}
\end{align*}
and finally, we obtain
\begin{align*}
\left\langle \boldsymbol{E}\boldsymbol{S}_{1}\boldsymbol{E}^{T},\boldsymbol{E}\boldsymbol{S}_{2}\boldsymbol{E}^{T}\right\rangle _{\boldsymbol{A}} & =\left\langle \boldsymbol{E}\boldsymbol{S}_{1}\boldsymbol{E}^{T}\boldsymbol{A}^{-1},\boldsymbol{A}^{-1}\boldsymbol{E}\boldsymbol{S}_{2}\boldsymbol{E}^{T}\right\rangle \\
 & =\text{Tr}\left\{ \boldsymbol{E}\boldsymbol{S}_{1}\boldsymbol{E}^{T}\boldsymbol{A}^{-1}\boldsymbol{E}\boldsymbol{S}_{2}\boldsymbol{E}^{T}\boldsymbol{A}^{-1}\right\} \\
 & =\text{Tr}\left\{ \boldsymbol{S}_{1}\boldsymbol{E}^{T}\boldsymbol{A}^{-1}\boldsymbol{E}\boldsymbol{S}_{2}\boldsymbol{E}^{T}\boldsymbol{A}^{-1}\boldsymbol{E}\right\} \\
 & =\text{Tr}\left\{ \boldsymbol{S}_{1}\boldsymbol{B}^{-1}\boldsymbol{S}_{2}\boldsymbol{B}^{-1}\right\} \\
 & =\left\langle \boldsymbol{S}_{1},\boldsymbol{S}_{2}\right\rangle _{\boldsymbol{B}}
\end{align*}
Proof of  condition (3):
Let $\boldsymbol{B}\in\mathcal{M}$ be an \ac{SPD} matrix with the following spectral decomposition:
\[
\boldsymbol{B}=\boldsymbol{M}\boldsymbol{\Lambda} \boldsymbol{M}^{T}
\]
Then, we have
\begin{align*}
\frac{d}{dt}\boldsymbol{B}^{t} & =\frac{d}{dt}\boldsymbol{M}\boldsymbol{\Lambda}^{t}\boldsymbol{M}^{T}\\
 & =\boldsymbol{M}\boldsymbol{\Lambda}^{t}\log\left(\boldsymbol{\Lambda}\right)\boldsymbol{M}^{T}\\
 & =\boldsymbol{M}\boldsymbol{\Lambda}^{t}\boldsymbol{M}^{T}\boldsymbol{M}\log\left(\boldsymbol{\Lambda}\right)\boldsymbol{M}^{T}\\
 & =\boldsymbol{B}^{t}\log\left(\boldsymbol{B}\right)
\end{align*}
Consider the geodesic $\varphi(t)$ from $\boldsymbol{B}$ to $\boldsymbol{A}$:
$$\varphi\left(t\right)=\boldsymbol{B}^{\frac{1}{2}}\left(\boldsymbol{B}^{-\frac{1}{2}}\boldsymbol{A}\boldsymbol{B}^{-\frac{1}{2}}\right)^{t}\boldsymbol{B}^{\frac{1}{2}}$$
Thus, its velocity at $t=0$ is given by
\begin{align}
\label{eq:phi0}
{\varphi}'\left(0\right)=\boldsymbol{B}^{\frac{1}{2}}\log\left(\boldsymbol{B}^{-\frac{1}{2}}\boldsymbol{A}\boldsymbol{B}^{\frac{1}{2}}\right)\boldsymbol{B}^{\frac{1}{2}}=\text{Log}_{\boldsymbol{B}}\left(\boldsymbol{A}\right)
\end{align}
and similarly, the velocity at $t=1$ is given by
\begin{align}
\begin{split}
\label{eq:phi1}
\varphi'\left(1\right) & =\boldsymbol{B}^{\frac{1}{2}}\left(\boldsymbol{B}^{-\frac{1}{2}}\boldsymbol{A}\boldsymbol{B}^{-\frac{1}{2}}\right)^{1}\log\left(\boldsymbol{B}^{-\frac{1}{2}}\boldsymbol{A}\boldsymbol{B}^{-\frac{1}{2}}\right)\boldsymbol{B}^{\frac{1}{2}}\\
 & =\boldsymbol{A}\boldsymbol{B}^{-\frac{1}{2}}\log\left(\boldsymbol{B}^{-\frac{1}{2}}\boldsymbol{A}\boldsymbol{B}^{-\frac{1}{2}}\right)\boldsymbol{B}^{\frac{1}{2}}\\
 & =-\boldsymbol{A}\boldsymbol{B}^{-\frac{1}{2}}\log\left(\boldsymbol{B}^{\frac{1}{2}}\boldsymbol{A}^{-1}\boldsymbol{B}^{\frac{1}{2}}\right)\boldsymbol{B}^{\frac{1}{2}}\\
 & =-\boldsymbol{A}\boldsymbol{B}^{-\frac{1}{2}}\log\left(\boldsymbol{B}^{\frac{1}{2}}\boldsymbol{A}^{-\frac{1}{2}}\boldsymbol{A}^{-\frac{1}{2}}\boldsymbol{B}\boldsymbol{A}^{-\frac{1}{2}}\boldsymbol{A}^{\frac{1}{2}}\boldsymbol{B}^{-\frac{1}{2}}\right)\boldsymbol{B}^{\frac{1}{2}}\\
 & \underbrace{=}_{\left(*\right)}-\boldsymbol{A}\boldsymbol{B}^{-\frac{1}{2}}\boldsymbol{B}^{\frac{1}{2}}\boldsymbol{A}^{-\frac{1}{2}}\log\left(\boldsymbol{A}^{-\frac{1}{2}}\boldsymbol{B}\boldsymbol{A}^{-\frac{1}{2}}\right)\boldsymbol{A}^{\frac{1}{2}}\boldsymbol{B}^{-\frac{1}{2}}\boldsymbol{B}^{\frac{1}{2}}\\
 & =-\boldsymbol{A}^{\frac{1}{2}}\log\left(\boldsymbol{A}^{-\frac{1}{2}}\boldsymbol{B}\boldsymbol{A}^{-\frac{1}{2}}\right)\boldsymbol{A}^{\frac{1}{2}}\\
 & =-\text{Log}_{\boldsymbol{A}}\left(\boldsymbol{B}\right)
 \end{split}
\end{align}
where in $\left(*\right)$ we pull out $\boldsymbol{V}=\boldsymbol{B}^{\frac{1}{2}}\boldsymbol{A}^{-\frac{1}{2}}$ and $\boldsymbol{V}^{-1}=\boldsymbol{A}^{\frac{1}{2}}\boldsymbol{B}^{-\frac{1}{2}}$ from the $\log$ , since it is a scalar function: $\log\left(\boldsymbol{V}\boldsymbol{P}\boldsymbol{V}^{-1}\right)=\boldsymbol{V}\log\left(\boldsymbol{P}\right)\boldsymbol{V}^{-1}$. 

Let $\boldsymbol{U}$ be the following unitary matrix
$$\boldsymbol{U}=\boldsymbol{A}^{-\frac{1}{2}}\boldsymbol{B}^{\frac{1}{2}}\left(\boldsymbol{B}^{-\frac{1}{2}}\boldsymbol{A}\boldsymbol{B}^{-\frac{1}{2}}\right)^{\frac{1}{2}}$$
Using $\boldsymbol{U}$, we can rewrite $\boldsymbol{E}$ as:
\begin{align}
\label{eq:U}
\boldsymbol{E}=\left(\boldsymbol{A}\boldsymbol{B}^{-1}\right)^{\frac{1}{2}}=\boldsymbol{B}^{\frac{1}{2}}\left(\boldsymbol{B}^{-\frac{1}{2}}\boldsymbol{A}\boldsymbol{B}^{-\frac{1}{2}}\right)^{\frac{1}{2}}\boldsymbol{B}^{-\frac{1}{2}}=\boldsymbol{A}^{\frac{1}{2}}\boldsymbol{U}\boldsymbol{B}^{-\frac{1}{2}}
\end{align}
Finally, by combining \eqref{eq:phi0}, \eqref{eq:phi1} and \ref{eq:U}, we have
\begin{align*}
\boldsymbol{E}{\varphi}'\left(0\right)\boldsymbol{E}^{T} & =\boldsymbol{A}^{\frac{1}{2}}\boldsymbol{U}\boldsymbol{B}^{-\frac{1}{2}}\boldsymbol{B}^{\frac{1}{2}}\log\left(\boldsymbol{B}^{-\frac{1}{2}}\boldsymbol{A}\boldsymbol{B}^{-\frac{1}{2}}\right)\boldsymbol{B}^{\frac{1}{2}}\boldsymbol{B}^{-\frac{1}{2}}\boldsymbol{U}^{T}\boldsymbol{A}^{\frac{1}{2}}\\
 & =\boldsymbol{A}^{\frac{1}{2}}\log\left(\boldsymbol{U}\boldsymbol{B}^{-\frac{1}{2}}\boldsymbol{A}\boldsymbol{B}^{-\frac{1}{2}}\boldsymbol{U}^{T}\right)\boldsymbol{A}^{\frac{1}{2}}\\
 & =\boldsymbol{A}^{\frac{1}{2}}\log\left(\boldsymbol{A}^{\frac{1}{2}}\boldsymbol{B}^{-1}\boldsymbol{A}^{\frac{1}{2}}\right)\boldsymbol{A}^{\frac{1}{2}}\\
 & =-\boldsymbol{A}^{\frac{1}{2}}\log\left(\boldsymbol{A}^{-\frac{1}{2}}\boldsymbol{B}\boldsymbol{A}^{-\frac{1}{2}}\right)\boldsymbol{A}^{\frac{1}{2}}\\
 & =-\text{Log}_{\boldsymbol{A}}\left(\boldsymbol{B}\right)\\
 & ={\varphi}'\left(1\right)
\end{align*}
\end{proof}

\section{Proof of theorem \ref{thm:EPE}}
\label{sec:proofThm_EPE}
\begin{proof}
\begin{align*}
\Psi\left(\boldsymbol{P}\right) & =\text{Exp}_{\boldsymbol{A}}\left(\Gamma_{\boldsymbol{B}\to \boldsymbol{A}}\left(\boldsymbol{S}\right)\right)\\
 & =\text{Exp}_{\boldsymbol{A}}\left(\boldsymbol{E}\text{Log}_{\boldsymbol{B}}\left(\boldsymbol{P}\right)\boldsymbol{E}^{T}\right)\\
 & =\text{Exp}_{\boldsymbol{A}}\left(\boldsymbol{E}\boldsymbol{B}^{\frac{1}{2}}\log\left(\boldsymbol{B}^{-\frac{1}{2}}\boldsymbol{P}\boldsymbol{B}^{-\frac{1}{2}}\right)\boldsymbol{B}^{\frac{1}{2}}\boldsymbol{E}^{T}\right)\\
 & =\boldsymbol{A}^{\frac{1}{2}}\exp\left(\boldsymbol{A}^{-\frac{1}{2}}\boldsymbol{E}\boldsymbol{B}^{\frac{1}{2}}\log\left(\boldsymbol{B}^{-\frac{1}{2}}\boldsymbol{P}\boldsymbol{B}^{-\frac{1}{2}}\right)\boldsymbol{B}^{\frac{1}{2}}\boldsymbol{E}^{T}\boldsymbol{A}^{-\frac{1}{2}}\right)\boldsymbol{A}^{\frac{1}{2}}\\
 & =\boldsymbol{A}^{\frac{1}{2}}\exp\left(\boldsymbol{U}\log\left(\boldsymbol{B}^{-\frac{1}{2}}\boldsymbol{P}\boldsymbol{B}^{-\frac{1}{2}}\right)\boldsymbol{U}^{T}\right)\boldsymbol{A}^{\frac{1}{2}}\\
 & =\boldsymbol{A}^{\frac{1}{2}}\boldsymbol{U}\exp\left(\log\left(\boldsymbol{B}^{-\frac{1}{2}}\boldsymbol{P}\boldsymbol{B}^{-\frac{1}{2}}\right)\right)\boldsymbol{U}^{T}\boldsymbol{A}^{\frac{1}{2}}\\
 & =\boldsymbol{A}^{\frac{1}{2}}\boldsymbol{U}\boldsymbol{B}^{-\frac{1}{2}}\boldsymbol{P}\boldsymbol{B}^{-\frac{1}{2}}\boldsymbol{U}^{T}\boldsymbol{A}^{\frac{1}{2}}\\
 & =\boldsymbol{A}^{\frac{1}{2}}\boldsymbol{A}^{-\frac{1}{2}}\boldsymbol{E}\boldsymbol{B}^{\frac{1}{2}}\boldsymbol{B}^{-\frac{1}{2}}\boldsymbol{P}\boldsymbol{B}^{-\frac{1}{2}}\boldsymbol{B}^{\frac{1}{2}}\boldsymbol{E}^{T}\boldsymbol{A}^{-\frac{1}{2}}\boldsymbol{A}^{\frac{1}{2}}\\
 & =\boldsymbol{E}\boldsymbol{P}\boldsymbol{E}^{T}
\end{align*}

where $\boldsymbol{U}=\boldsymbol{A}^{-\frac{1}{2}}\boldsymbol{E}\boldsymbol{B}^{\frac{1}{2}}$ is a unitary matrix, and therefore can be pulled out of the scalar exp function..
\end{proof}

\section{Proof of proposition \ref{prop:Invariant}}
\label{sec:proofProp_Invariant}
\begin{figure*}[t]
\begin{align}
\begin{split}
\label{eq:MidGeoPT}
\boldsymbol{B}_{2}^{\frac{1}{2}}\left(\boldsymbol{B}_{2}^{-\frac{1}{2}}\boldsymbol{A}_{2}\boldsymbol{B}_{2}^{-\frac{1}{2}}\right)^{\frac{1}{2}}\boldsymbol{B}_{2}^{\frac{1}{2}} & =\boldsymbol{B}_{2}^{\frac{1}{2}}\left(\left(\boldsymbol{E}\boldsymbol{B}_{1}\boldsymbol{E}^{T}\right)^{-\frac{1}{2}}\boldsymbol{E}\boldsymbol{A}_{1}\boldsymbol{E}^{T}\left(\boldsymbol{E}\boldsymbol{B}_{1}\boldsymbol{E}^{T}\right)^{-\frac{1}{2}}\right)^{\frac{1}{2}}\boldsymbol{B}_{2}^{\frac{1}{2}}\\
 & =\boldsymbol{B}_{2}^{\frac{1}{2}}\left(\underbrace{\left(\boldsymbol{E}\boldsymbol{B}_{1}\boldsymbol{E}^{T}\right)^{\frac{1}{2}}\boldsymbol{E}^{-T}\boldsymbol{B}_{1}^{-\frac{1}{2}}}_{\boldsymbol{K}}\boldsymbol{B}_{1}^{-\frac{1}{2}}\boldsymbol{A}_{1}\boldsymbol{B}_{1}^{-\frac{1}{2}}\underbrace{\boldsymbol{B}_{1}^{-\frac{1}{2}}\boldsymbol{E}^{-1}\left(\boldsymbol{E}\boldsymbol{B}_{1}\boldsymbol{E}^{T}\right)^{\frac{1}{2}}}_{\boldsymbol{K}^{T}}\right)^{\frac{1}{2}}\boldsymbol{B}_{2}^{\frac{1}{2}}\\
 & =\boldsymbol{B}_{2}^{\frac{1}{2}}\left(\boldsymbol{E}\boldsymbol{B}_{1}\boldsymbol{E}^{T}\right)^{\frac{1}{2}}\boldsymbol{E}^{-T}\boldsymbol{B}_{1}^{-\frac{1}{2}}\left(\boldsymbol{B}_{1}^{-\frac{1}{2}}\boldsymbol{A}_{1}\boldsymbol{B}_{1}^{-\frac{1}{2}}\right)^{\frac{1}{2}}\boldsymbol{B}_{1}^{-\frac{1}{2}}\boldsymbol{E}^{-1}\left(\boldsymbol{E}\boldsymbol{B}_{1}\boldsymbol{E}^{T}\right)^{\frac{1}{2}}\boldsymbol{B}_{2}^{\frac{1}{2}}\\
 & =\boldsymbol{E}\boldsymbol{B}_{1}\boldsymbol{E}^{T}\boldsymbol{E}^{-T}\boldsymbol{B}_{1}^{-\frac{1}{2}}\left(\boldsymbol{B}_{1}^{-\frac{1}{2}}\boldsymbol{A}_{1}\boldsymbol{B}_{1}^{-\frac{1}{2}}\right)^{\frac{1}{2}}\boldsymbol{B}_{1}^{-\frac{1}{2}}\boldsymbol{E}^{-1}\boldsymbol{E}\boldsymbol{B}_{1}\boldsymbol{E}^{T}\\
 & =\boldsymbol{E}\boldsymbol{B}_{1}^{\frac{1}{2}}\left(\boldsymbol{B}_{1}^{-\frac{1}{2}}\boldsymbol{A}_{1}\boldsymbol{B}_{1}^{-\frac{1}{2}}\right)^{\frac{1}{2}}\boldsymbol{B}_{1}^{\frac{1}{2}}\boldsymbol{E}^{T}
\end{split}
\end{align}
\end{figure*}

\begin{proof}
Let $\boldsymbol{E}_{1}=\left(\boldsymbol{A}_{1}\boldsymbol{B}_{1}^{-1}\right)^{\frac{1}{2}}=\boldsymbol{B}_{1}^{\frac{1}{2}}\left(\boldsymbol{B}_{1}^{-\frac{1}{2}}\boldsymbol{A}_{1}\boldsymbol{B}_{1}^{-\frac{1}{2}}\right)\boldsymbol{B}_{1}^{-\frac{1}{2}}$ and $\boldsymbol{E}_{2}=\left(\boldsymbol{A}_{2}\boldsymbol{B}_{2}^{-1}\right)^{\frac{1}{2}}=\boldsymbol{B}_{2}^{\frac{1}{2}}\left(\boldsymbol{B}_{2}^{-\frac{1}{2}}\boldsymbol{A}_{2}\boldsymbol{B}_{2}^{-\frac{1}{2}}\right)\boldsymbol{B}_{2}^{-\frac{1}{2}}$.

Since
\[
\Gamma\left(\Gamma_{\boldsymbol{B}_{1}\to \boldsymbol{A}_{1}}\left(\boldsymbol{S}\right)\right)=\boldsymbol{E}\boldsymbol{E}_{1}\boldsymbol{S}\boldsymbol{E}_{1}^{T}\boldsymbol{E}^{T}\,,
\]
 and 
\[
\Gamma_{\boldsymbol{B}_{2}\to \boldsymbol{A}_{2}}\left(\Gamma\left(\boldsymbol{S}\right)\right)=\boldsymbol{E}_{2}\boldsymbol{E}\boldsymbol{S}\boldsymbol{E}^{T}\boldsymbol{E}_{2}^{T}\,,
\]
it is enough to show that
\[
\boldsymbol{E}_{2}\boldsymbol{E}=\boldsymbol{E}\boldsymbol{E}_{1}.
\]

First, note that $\boldsymbol{K}=\boldsymbol{B}_{2}^{\frac{1}{2}}\boldsymbol{E}^{-T}\boldsymbol{B}_{1}^{-\frac{1}{2}}$ is unitary:
\begin{align*}
\boldsymbol{K}\boldsymbol{K}^{T} & =\boldsymbol{B}_{2}^{\frac{1}{2}}\boldsymbol{E}^{-T}\boldsymbol{B}_{1}^{-\frac{1}{2}}\left(\boldsymbol{B}_{2}^{\frac{1}{2}}\boldsymbol{E}^{-T}\boldsymbol{B}_{1}^{-\frac{1}{2}}\right)^{T}\\
 & =\boldsymbol{B}_{2}^{\frac{1}{2}}\boldsymbol{E}^{-T}\boldsymbol{B}_{1}^{-\frac{1}{2}}\boldsymbol{B}_{1}^{-\frac{1}{2}}\boldsymbol{E}^{-1}\boldsymbol{B}_{2}^{\frac{1}{2}}\\
 & =\boldsymbol{B}_{2}^{\frac{1}{2}}\boldsymbol{E}^{-T}\boldsymbol{B}_{1}^{-1}\boldsymbol{E}^{-1}\boldsymbol{B}_{2}^{\frac{1}{2}}\\
 & =\boldsymbol{B}_{2}^{\frac{1}{2}}\boldsymbol{B}_{2}^{-1}\boldsymbol{B}_{2}^{\frac{1}{2}}\\
 & =\boldsymbol{I}\,.
\end{align*}
Now, we have \eqref{eq:MidGeoPT}.
Finally, we have
\begin{align*}
\boldsymbol{E}_{2}\boldsymbol{E} & =\boldsymbol{B}_{2}^{\frac{1}{2}}\left(\boldsymbol{B}_{2}^{-\frac{1}{2}}\boldsymbol{A}_{2}\boldsymbol{B}_{2}^{-\frac{1}{2}}\right)^{\frac{1}{2}}\boldsymbol{B}_{2}^{-\frac{1}{2}}\boldsymbol{E}\\
 & =\boldsymbol{B}_{2}^{\frac{1}{2}}\left(\boldsymbol{B}_{2}^{-\frac{1}{2}}\boldsymbol{A}_{2}\boldsymbol{B}_{2}^{-\frac{1}{2}}\right)^{\frac{1}{2}}\boldsymbol{B}_{2}^{\frac{1}{2}}\boldsymbol{B}_{2}^{-1}\boldsymbol{E}\\
 & =\boldsymbol{E}\boldsymbol{B}_{1}^{\frac{1}{2}}\left(\boldsymbol{B}_{1}^{-\frac{1}{2}}\boldsymbol{A}_{1}\boldsymbol{B}_{1}^{-\frac{1}{2}}\right)^{\frac{1}{2}}\boldsymbol{B}_{1}^{\frac{1}{2}}\boldsymbol{E}^{T}\boldsymbol{B}_{2}^{-1}\boldsymbol{E}\\
 & =\boldsymbol{E}\boldsymbol{B}_{1}^{\frac{1}{2}}\left(\boldsymbol{B}_{1}^{-\frac{1}{2}}\boldsymbol{A}_{1}\boldsymbol{B}_{1}^{-\frac{1}{2}}\right)^{\frac{1}{2}}\boldsymbol{B}_{1}^{-\frac{1}{2}}\\
 & =\boldsymbol{E}\boldsymbol{E}_{1}\,.
\end{align*}
\end{proof}

\section{ }
\label{sec:proofCommute}
In this appendix we proof that if the identity matrix $\boldsymbol{I}$ is on the geodesic $\varphi$ between $\boldsymbol{A}$ and  $\boldsymbol{B}$, then, the matrices $\boldsymbol{A}$ and $\boldsymbol{B}$ commute and they have the same eigenvectors.
\begin{proof}
if $\boldsymbol{I}$ is on the geodesic $\varphi(t)$, then there exist some $t_0\in(0,1)$ (if $t_0=0$ or $t_0=1$ the result is trivial) such that:
$$\varphi\left(t_{0}\right)=\boldsymbol{A}^{\frac{1}{2}}\left(\boldsymbol{A}^{-\frac{1}{2}}\boldsymbol{B}\boldsymbol{A}^{-\frac{1}{2}}\right)^{t_0}\boldsymbol{A}^{\frac{1}{2}}=\boldsymbol{I}$$
Now, consider $\boldsymbol{A}=\boldsymbol{V}\boldsymbol{\Lambda}\boldsymbol{V}^{T}$ the eigenvalue-decomposition of $\boldsymbol{A}$. By multiplying both from the right and the left, we have
\begin{align*}
\boldsymbol{A}^{\frac{1}{2}}\left(\boldsymbol{A}^{-\frac{1}{2}}\boldsymbol{B}\boldsymbol{A}^{-\frac{1}{2}}\right)^{t_{0}}\boldsymbol{A}^{\frac{1}{2}} & =\boldsymbol{I}\\
\left(\boldsymbol{A}^{-\frac{1}{2}}\boldsymbol{B}\boldsymbol{A}^{-\frac{1}{2}}\right)^{t_{0}} & =\boldsymbol{A}^{-1}\\
\end{align*}
By raising to the power of $\frac{1}{t_0}$, and then multiplying both from the right and the left again, we have
\begin{align*}
\boldsymbol{A}^{-\frac{1}{2}}\boldsymbol{B}\boldsymbol{A}^{-\frac{1}{2}} & =\boldsymbol{A}^{-\frac{1}{t_{0}}}\\
\boldsymbol{B} & =\boldsymbol{A}^{1-\frac{1}{t_{0}}}\\
\boldsymbol{B} & =\boldsymbol{B}\boldsymbol{\Lambda}^{1-\frac{1}{t_{0}}}\boldsymbol{V}^{T}
\end{align*}

Thus, $\boldsymbol{B}$ has the same eigenvectors as $\boldsymbol{A}$ and they commute
\begin{align*}
\boldsymbol{A}\boldsymbol{B} & =\boldsymbol{V}\boldsymbol{\Lambda}\boldsymbol{V}^{T}\boldsymbol{V}\boldsymbol{\Lambda}^{1-\frac{1}{t_{0}}}\boldsymbol{V}^{T}\\
 & =\boldsymbol{V}\boldsymbol{\Lambda}\boldsymbol{\Lambda}^{1-\frac{1}{t_{0}}}\boldsymbol{V}^{T}\\
 & =\boldsymbol{V}\boldsymbol{\Lambda}^{1-\frac{1}{t_{0}}}\boldsymbol{V}^{T}\boldsymbol{V}\boldsymbol{\boldsymbol{\Lambda}V}^{T}\\
 & =\boldsymbol{B}\boldsymbol{A}
\end{align*}
\end{proof}

\section{Riemannian mean algorithm}
\begin{algorithm}[tbh]
\label{alg:RiemannianMean}
\textbf{\uline{Input}}\textbf{: }a set of SPD matrices $\left\{ \boldsymbol{P}_{i}\in\mathcal{M}\right\} _{i=1}^{N}$.

\textbf{\uline{Output}}\textbf{:} the Riemannian mean matrix $\overline{\boldsymbol{P}}$.
\begin{enumerate}
\item Compute the initial term $\overline{\boldsymbol{P}}=\frac{1}{N}\sum_{i=1}^{N}\boldsymbol{P}_{i}$
\item \textbf{do }
\begin{enumerate}
\item Compute the Euclidean mean in the tangent space: $\overline{\boldsymbol{S}}=\frac{1}{N}\sum_{i=1}^{N}\text{Log}_{\overline{P}}\left(\boldsymbol{P}_{i}\right)$
\item Update $\overline{\boldsymbol{P}}=\text{Exp}_{\overline{P}}\left(\overline{\boldsymbol{S}}\right)$
\item \textbf{while $\left\Vert \overline{\boldsymbol{S}}\right\Vert _{F}>\epsilon$} where $\left\Vert \cdot \right\Vert _{F}$ is the Frobenius norm.
\end{enumerate}
\end{enumerate}
\caption{Riemannian mean for \ac{SPD} matrices as presented in \cite{barachant2013classification}}
\end{algorithm}



\bibliographystyle{IEEEtran}
\bibliography{Refs}



\end{document}